\documentclass[10pt]{article}
\usepackage{natbib}
\usepackage{graphicx}
\usepackage{amsmath}
\usepackage{authblk}
\usepackage{amsfonts}
\usepackage{amsthm}
\usepackage{enumitem}
\usepackage[font={footnotesize}]{caption}
\usepackage{framed}
\usepackage{url}
\usepackage{placeins}
\usepackage[paper=a4paper,left=30mm,right=30mm,top=25mm,bottom=25mm]{geometry}

\renewcommand{\P}{\mathbb{P}}
\renewcommand{\d}{\mathrm{d}}

\DeclareMathOperator{\argmin}{\mathrm{argmin}}
\DeclareMathOperator{\E}{\mathbb{E}}

\DeclareMathOperator{\R}{\mathbb{R}}

\DeclareMathOperator{\B}{\mathcal{B}}

\DeclareMathOperator{\bfZ}{\mathbf{Z}}

\DeclareMathOperator{\bfG}{\mathbf{G}}
\DeclareMathOperator{\bfF}{\mathbf{F}}
\DeclareMathOperator{\bfU}{\mathbf{U}}
\DeclareMathOperator{\bfI}{\mathbf{I}}

\DeclareMathOperator{\Pois}{\mathrm{Pois}}
\DeclareMathOperator{\PRM}{\mathrm{PRM}}

\DeclareMathOperator{\bfX}{\mathbf{X}}
\DeclareMathOperator{\bfY}{\mathbf{Y}}
\DeclareMathOperator{\bfE}{\mathbf{E}}

\DeclareMathOperator{\bfH}{\mathbf{H}}

\DeclareMathOperator{\bfN}{\mathbf{N}}
\DeclareMathOperator{\N}{\mathbb{N}}
\DeclareMathOperator{\bfeta}{\mathbf\eta}
\DeclareMathOperator{\bfw}{\mathbf w}

\renewcommand{\vec}{\mathrm{vec}}

\renewcommand{\(}{\left(}
\renewcommand{\)}{\right)}
\renewcommand{\[}{\left[}
\renewcommand{\]}{\right]}

\renewcommand{\P}{\mathbb{P}}
\renewcommand{\d}{\mathrm{d}}
\newcommand{\PA}{\mathrm{PA}}

\newcommand{\AN}{\mathrm{AN}}

\newtheorem{definition}{Definition}
\newtheorem{theorem}{Theorem}

\newtheorem{remark}{Remark}
\newtheorem{proposition}{Proposition}
\newtheorem{algorithm}{Algorithm}

\theoremstyle{definition}

\title{Hawkes graphs\thanks{This work was
supported by RiskLab Zurich and the Swiss Finance Institute.
}} 
\date{This version: \today}
\author{
Paul Embrechts, Matthias Kirchner\thanks{
 Corresponding author: matthias.kirchner@math.ethz.ch}}
\affil{RiskLab, 
Department of Mathematics,
ETH Zurich,
R{\"a}mistrasse 101,
8092 Zurich,
Switzerland
}
\begin{document}

\maketitle
\begin{abstract}
This paper introduces the Hawkes skeleton and the Hawkes graph. These objects summarize the branching structure of a multivariate Hawkes point process in a compact, yet meaningful way. We demonstrate how graph-theoretic vocabulary (`ancestor sets', `parent sets', `connectivity', `walks', `walk weights', \dots) is very convenient for the discussion of multivariate Hawkes processes. For example, we reformulate the classic eigenvalue-based subcriticality criterion of multitype branching processes in graph terms.
Next to these more terminological contributions, we show how the graph view may be used for the specification and estimation of Hawkes models from large, multitype event streams. Based on earlier work, we give a nonparametric statistical procedure to estimate the Hawkes skeleton and the Hawkes graph from data. We show how the graph estimation may then be used for specifying and fitting parametric Hawkes models. Our estimation method avoids the a priori assumptions on the model from a straighforward MLE-approach and is numerically more flexible than the latter. Our method has two tuning parameters: one controlling numerical complexity, the other one controlling the sparseness of the estimated graph. A simulation study confirms that the presented procedure works as desired. We pay special attention to computational issues in the implementation. This makes our results applicable to high-dimensional event-stream data, such as dozens of event streams and thousands of events per component.  
\end{abstract}

\section{Introduction}
This paper discusses the specification and estimation of multivariate Hawkes point process models from large, multitype event-stream datasets such as neural spike-trains, internet search-queries, or limit-order-book data in high-frequency finance. Our approach uses the notion of a Hawkes skeleton and a Hawkes graph\footnote{Note that the term `{Hawkes graph}' has already been introduced for the graph representation of a specific finite group; see \citet{hawkes68}. Neither the author of the latter paper, \emph{T.}\ Hawkes, nor its content has anything to do with our notion of a Hawkes graph.
}. We demonstrate how these concepts are fertile beyond statistical estimation.

The Hawkes process was introduced in \cite{hawkes71a, hawkes71b} as a stationary point process on $\R$ whose points are assigned to a finite number of types.
The (stochastic) intensity of a Hawkes process depends on the past of the process itself: given the occurrence of an event, the intensities---the expected mean number of events per time unit and event type---typically jump upwards and then decay. This structure can alternatively be represented as a multitype branching-process with immigration; see \cite{hawkes74}. The crucial parameters of a Hawkes model are the \emph{excitement functions} or, emphasizing the branching interpretation, the \emph{reproduction intensities} that govern these self- and crosseffects.
For a textbook reference that covers many aspects of the Hawkes process, see \cite{daley03}. Maximum likelihood estimation of Hawkes processes has been treated in \cite{ogata88} covering calibration issues and introducing a computationally beneficial recursive method for the exponential decay case. \cite{liniger09} deals especially with the construction of the multivariate and marked case.\\ 
\par In the present paper, we formally introduce the \emph{Hawkes graph}. The Hawkes graph summarizes the branching structure of a multitype Hawkes point process as a directed graph with weighted vertices and edges. The vertices represent the possible event-types of the corresponding Hawkes process; an edge $(i,j)$ denotes nonzero excitement from event-type $i$ to event-type $j$. The vertex weights are the corresponding immigration intensities; the weight of an edge $(i,j)$ is the expected number of type-$j$ children events that an type-$i$ event generates.  The \emph{Hawkes skeleton} is the Hawkes graph  disregarding the weights. The network view on Hawkes processes has been considered in \cite{song13}, \cite{delattre15},  \cite{bacry15}, and \cite{hall16}. The graph terminology is convenient to describe many relevant aspects of multivariate Hawkes processes such as `ancestor and parent sets', `paths', `path weights', `feedback', `cascades', or `connectivity'. The graph representation of a Hawkes process also provides additional theoretical insight. For example, in Theorem~\ref{prop:graph_subcriticality}, we give a graph-based criterion for subcriticality which is equivalent to the usual spectral-radius based criterion on the branching matrix. Furthermore, the graph approach turns out to be helpful for the estimation of multivariate Hawkes processes.

Concerning Hawkes process estimation, we see three main problems with the standard parametric likelihood approach. First of all, it uses many unjustified assumptions on the shape of the reproduction intensities. Secondly, the distribution of the MLE-estimator is (in general) not known. In particular, the likelihood approach does not provide tests to decide whether excitement from one event type to another exists \emph{at all}. Finally, there are numerical issues that make it difficult to apply MLE in a straightforward way with large, high-dimensional event-stream datasets. 
\par Our approach leaves the choice of the excitement functions open to the very last.  We apply an estimation procedure developed in \citet{kirchner16d}. This procedure is based on a limit-representation of the Hawkes process studied in \citet{kirchner16c}: we discretize the original process and interpret it as an autoregressive model of bin-counts. The latter is statistically estimated using conditional least-squares. In this setup, the asymptotic distribution of the resulting estimators can be obtained. This opens the door to testing. Our procedure is numerically more robust than the standard MLE approach. However, for high-dimensional data our procedure cannot be applied in a straightforward manner either. This is why, in combination with the concept of a Hawkes skeleton and graph, we tackle the numerical difficulties by the following three-step algorithm: 
\begin{enumerate}
\item Given a large multitype event-stream dataset, we first apply a specific testing scheme to decide whether there is \emph{any} effect from a specific event type to any other event type. The test result yields the \emph{Hawkes-skeleton estimate}. In this first step, we use a parameter allowing us to tune for a \emph{very coarse discretization}; this keeps the computational complexity under control. Despite the resulting discretization error, this approach typically yields a \emph{superset} of the true edge set. Under the paradigm that the graph of the true underlying multivariate Hawkes model is typically sparse, this estimated superset is still sparse. 
\item In a second step, we estimate the \emph{Hawkes graph given the skeleton estimate}. The Hawkes graph \emph{quantifies} the remaining excitement effects. The sparseness of the estimated Hawkes-skeleton from (i) reduces the complexity of the estimation problem considerably: there are only few excitements left to estimate and there are fewer `explanatory types' per event type, namely the estimated parent sets. Consequently, we may now choose a much finer discretization parameter and thus retrieve more precise edge and vertex weight estimates---including confidence intervals for all estimated values. 
\item As a by-product, the calculations in (ii) yield estimates for the values of the nonzero excitement-functions on a finite equidistant grid. We exploit these estimation results graphically to choose appropriate parametric function-families. Finally, we fit the chosen parametric functions to the corresponding estimates by a non-linear least-squares method. This yields parameter estimates for parametric Hawkes models.
\end{enumerate}
The multistep-procedure described above also works in a high-dimensional setting (such as dozens of event streams and thousands of events per component); the approach can be implemented in a straightforward way. \\

\par The paper is organized as follows: In Section 2, we give definitions and discuss graph attributes that are relevant for the description of multivariate Hawkes processes. In particular, we give results that clarify what kind of information on the Hawkes process a Hawkes graph encodes. In Section 3, we cite earlier results that allow for nonparametric estimation of Hawkes processes. We apply these methods to estimate the Hawkes skeleton and the Hawkes graph. Finally, we show how parametric families for the remaining nonzero reproduction intensities may be specified and calibrated. For an illustration of the new concepts introduced, we present a simulation study in Section 4. In Section 5, we conclude with directions for further research. 

\section{Definitions}\label{definitions}
In this section, we recall the branching construction of a multivariate Hawkes process as well as basic graph terminology. After this, we introduce the Hawkes skeleton as well as the Hawkes graph. The graph representation summarizes the branching structure of a Hawkes process in a compact and insightful manner. 

\subsection{Multivariate Hawkes processes}\label{multivariate_def:hawkes_processes}
Throughout the paper, let $(\Omega, \P, \mathcal{F})$ be a complete probability space rich enough to carry all random variables involved. We give a constructive definition of the Hawkes process that emphasizes the branching structure.  For a similar construction; see \citet{hawkes74} or Chapter 4 in \citet{liniger09}. The building blocks are  Poisson random-measures on $\R$ endowed with the Borel $\sigma$-algebra $\mathcal{B}(\R)$.
\begin{definition}
Let $\lambda: \R \to \R_{\geq 0}$ be a locally integrable function. 
We say that  $M$ is a \emph{Poisson random-measure} on $(\R, \B(\R))$ with \emph{intensity function} $\lambda$  whenever the following two conditions hold:
\begin{enumerate}
\item  $M(B) \sim \Pois\( \int_B \lambda(s)\d s\), \ B\in\B(\R)$.
\item  If $B_1,B_2,\dots,B_n\in\B(\R)$ with $B_i \cap B_j = \emptyset,\, i\neq j$, then $M(B_1), M(B_2), \dots, M(B_n) $ are mutually independent.
\end{enumerate}
We write 
$M\sim \PRM(\lambda \d s)$.
\end{definition}
In the definition above we use the convention that 
$X\sim \Pois\(0\) :\Leftrightarrow X \equiv0,$ a.s. and 
$X\sim \Pois\(\infty\) :\Leftrightarrow X \equiv \infty,$ a.s. \\ 
\par A multitype Hawkes process is a model for the occurrence of events on $\R$, where the events are assigned to a finite number of types. The different event-types are represented as (in general dependent) random counting measures. For each event type, there is an immigration process. Each immigrant event independently generates a family. These families consist of cascades of Poisson random measures. A Hawkes process is the superposition of all such families. We formalize this construction in the definitions below. To emphasize the intuition behind the names of  immigrants, generations, and families, we use the somewhat unusual letters $\bfI$, $\bfG$, and $\bfF$ for the corresponding processes.

\begin{definition}\label{hawkes_family}
Let $d\in\N$ and $[d] := \{1,2,\dots, d\}$.
\begin{enumerate}
\item For $(i,j)\in[d]^2$, define \emph{branching coefficients} $a_{i,j}\geq 0$, \emph{displacement densities} $w_{i,j}$ supported on $\R_{\geq0}$, 
\emph{reproduction intensities} $h_{i,j} := a_{i,j}w_{i,j}$, and \emph{reproduction processes} $\xi^{(i,j)}_{t}(\cdot) := \xi^{(i,j)}(\cdot - t) \sim \PRM(h_{i,j}\d s)$, $t\in\R$, mutually independent over $(i,j,t) \in[d]^2\times\R$. 

\item For $i_0\in[d]$ and $g\in\N_0$, define the $g$-th \emph{generation process} (generated by a type-$i_0$ event at time zero) as the $d$-tuple of random counting measures
$\bfG^{(i_0,g)}:= \(G^{(i_0,g)}_1,\dots,\, G_d^{(i_0,g)}\)$ by
\begin{align}
G^{(i_0,0)}_j(B) &:= 1_{\{j = i_0\}}\delta_{0}(B),\quad B\in\mathcal{B}(\R),\, j\in[d],\nonumber\\
 G^{(i_0,g)}_j(B)&:= \sum\limits_{i = 1}^d \int_{\R}\xi^{(i,j)}_{t}(B)  G^{(i_0,g-1)}_i(\d t),\quad B\in\mathcal{B}(\R),\, j\in[d],\, g\in\N.\label{G_j}
\end{align}
\item For $i_0\in[d]$, define the \emph{Hawkes family} (generated by a type-$i_0$ event at time zero) as the $d$-tuple of random counting measures
$$\bfF^{(i_0)}= \sum_{g\geq 0} \bfG^{(i_0,g)}.$$
\end{enumerate}
\end{definition}
The branching structure of a Hawkes family is encoded in recursion \eqref{G_j}. 
Note that the points of a Hawkes family actually form a \emph{multitype branching random walk}; see \citet{shi15}. The following definition clarifies how the Hawkes family process is related to the prototypic branching process, the Galton--Watson process:
\begin{definition}\label{embedded_generation_process}
For $i_0 \in[d],$ let $\bfF^{(i_0)}$ be a Hawkes family and let $\{\bfG^{(i_0,g)}\}_{g\in\N_0}$ be the corresponding generation processes constructed in Definition~\ref{hawkes_family} above.
For $g\in\N_0$, define 
$$
\bfY^{(i_0)}_g := \(Y^{(i_0)}_{g,1}, Y^{(i_0)}_{g,2}, \dots, Y^{(i_0)}_{g,d}\),  \quad \text{where, for } j\in[d],\quad Y^{(i_0)}_{g,j} :=  G^{(i_0,g)}_j(\R). 
$$
We call $(\bfY^{(i_0)}_g)_{g\in\N_0}$  the \emph{embedded generation process} of the Hawkes family $\bfF^{(i_0)}$.
\end{definition}
The embedded generation process $(\bfY^{(i_0)}_g)$ of a Hawkes family is a multitype Galton--Watson process. A multitype Galton--Watson process models the size of a population with individuals of $d$ types, where each individual is alive during exactly one time unit; see Section~2.3 in \citet{jagers05}. The embedded generation process starts with a single type-$i_0$ individual in generation 0 and, for $g\in\N$, each type-$i$ individual in generation $g-1$ gives offspring to $\mathrm{Pois}(a_{i,j})$ $a_{i,j} = \int h_{i,j} \d t$) type-$j$ individuals of type $j$ in generation $g$. This is why $a_{i,j} ,\,  (i,j)\in[ d]^2,$ are called \emph{branching coefficients} and why the matrix $A:=(a_{i,j})\in\R_{\geq 0}$ is called \emph{branching matrix}. 
\begin{proposition}\label{prop:expectation}
Let $A$ be the branching matrix of Hawkes families $\bfF^{(i_0)},\, i_0\in[d]$, respectively, of the corresponding embedded generation processes $(\bfY^{(i_0)}_g),\, i_0\in[d]$. Then we have that
\begin{align}
\E F^{(i_0)}_j(\R) =\sum_{g\geq 0} \E Y^{(i_0)}_{g,j} < \infty,\quad (i_0, j)\in[d]^2,
\end{align}
if and only if the spectral radius of $A$ is strictly less than 1. In this case, $(1_{d\times d} - A)$ is invertible and $(\E F^{(i_0)}_j(\R))_{(i_0,j)\in[d]^2} =  (1_{d\times d} - A)^{-1}$.
\end{proposition}
\begin{proof}
Using 
$$
\E \bfY^{(i_0)}_{0} =  \bfY^{(i_0)}_{0} = (0,\dots, 0, \underbrace{1}_{\text{$i_0$-th entry}}, 0, \dots, 0) \text{  and  }\E \bfY^{(i_0)}_{g}=\E \bfY^{(i_0)}_{g-1} A,\, g\in\N,\, i_0\in[d],
$$
 it follows by induction that
$(\E Y^{(i_0)}_{j,g})_{(i_0, j)\in[d]^2} = A^g,\, g\in\N_0$.
 By Fubini's theorem, we then get that
 $(\E F^{(i_0)}_j(\R))_{(i_0, j)\in[d]} = \sum_{g\geq 0} (\E Y^{(i_0)}_{g,j})_{(i_0, j)\in[d]}
= \sum_{g \geq 0} A^g$.
Given its entries are finite, the limit matrix $\sum_{g \geq 0} A^g$ is calculated like the limit of a real-valued converging geometric series. The equivalence in Proposition~\ref{prop:expectation} follows from the fact that
\begin{align}
\sum_{g = 0}^\infty A^g \text{ converges}\quad \Leftrightarrow \quad \max\big\{|\lambda|:\, \lambda \text{ eigenvalue of } A\big\}< 1,\quad \text{for $A\in\R^{d\times d}$.}\label{eq:eigenvalue}
\end{align}
A detailed proof for \eqref{eq:eigenvalue} can be found in \citet{watson15}. 
\end{proof}
In particular, we get from Proposition~\ref{prop:expectation} that a Hawkes family whose branching matrix satisfies \eqref{eq:eigenvalue} consists of an almost surely finite number of points. 
\begin{definition}\label{def:hawkes_process}
 Let $\bfI = (I_1, I_2, \dots, I_d)$ be a \emph{Hawkes immigration process} with $I_{i_0}\sim$ PRM($\eta_{i_0}\d s$), $i_0\in[d]$, independent, where $\eta_{i_0}\geq 0, \, i_0\in [d],$ are (constant) \emph{immigration intensities}.  Furthermore, let $\bfF^{(i_0)}_t(\cdot):= \bfF^{(i_0,t)}(\cdot - t),\, t\in\R$, where $\bfF^{(i_0,t)},\, t\in\R, \,i_0 \in[d],$ are independent copies of the generic Hawkes family processes
$\bfF^{(i_0)}$ from Definition~\ref{hawkes_family} above---also independent from the immigration process $\bfI$. 
Set
$$
\bfN (B) :=\big(N_1(B), \dots, N_d(B)\big):= \sum\limits_{i_0 = 1}^d\int_{\R}\bfF^{(i_0)}_t(B)\,I_{i_0}(\d t),\quad B\in\mathcal{B}(\R).
$$
The $d$-tuple of random counting measures $\bfN$ is a \emph{$d$-type Hawkes process}. If $N_i(\{T\}) =1$, for some $i \in[d]$, we say that $T$ is a \emph{type-$i$ event}  or, synonymously,  an \emph{event in component $i$}.  The Hawkes process $\bfN$ is \emph{subcritical} if the corresponding embedded generation processes are subcritical, i.e., if the spectral radius of their branching matrix is strictly smaller than 1. 
\end{definition}
From \citet{hawkes74} we have that, in the subcritical case, a Hawkes process $\bfN$, constructed as in Definitions~\ref{hawkes_family} and \ref{def:hawkes_process}, is a stationary solution to the system of
implicit equations
\begin{eqnarray}
{\Lambda}_j(t)&:=&\lim\limits_{\delta\downarrow 0} \frac{1}{\delta}\E\bigg[{N}_j\big((t, t+\delta]\big)\Big|\sigma\Big(\bfN\big((a,b]\big),\, a<b\leq t\Big)\bigg]\nonumber\\
&=&
\bfeta_j + \sum\limits_{i = 1}^d \int\limits_{-\infty}^{t}h_{i,j}(t-s) N_i
\left(\d s\right),\quad t\in\R,\,  j\in[d].\label{multivariate_intensity}
\end{eqnarray}
We call $\mathbf{\Lambda}(t) := ({\Lambda}_1(t),{\Lambda}_2(t),\dots,{\Lambda}_d(t))$ the \emph{conditional intensity} of $\bfN$. In terms of intensities, the value of a reproduction intensity at time $t$, $h_{i,j}(t)$, denotes the effect of an event $T^{(i)}$ in component $i$  on the intensity of component $j$ at time $T^{(i)}+t$. 
\begin{remark}\label{ij}
In most work on Hawkes processes, including the original introductions \citep{hawkes71a, hawkes71b} and also including \citep{kirchner16d}, the function $h_{i,j}$ models the excitement \emph{from component $j$ on component $i$}. This somewhat counter-intuitive notation stems from the linear algebra used when writing \eqref{multivariate_intensity} with matrix multiplication. In the present graph-driven work, `$a_{i,j}$', `$w_{i,j}$', `$h_{i,j}$', and `$(i,j)\in\mathcal{E}$' \emph{all refer to the effect from component $i$ on component $j$.}
\end{remark}

\subsection{Hawkes skeleton and Hawkes graph}
We interpret the branching structure of the Hawkes process in terms of `causality'. The overall goal of causality research is to describe dependencies in a directed manner---rather than applying commutative concepts such as correlation; see~\cite{pearl09} for a recent overview. The notion of causality is subtle. For Hawkes processes, however, the use of the term seems justified. Indeed, in the context of event streams, things cannot become much more `causal' than in the recurrent parent/children relation of a branching process: if we delete an event in the branching construction from the definitions in Section~\ref{multivariate_def:hawkes_processes} above, its offspring vanishes. So---without discussing causality formally---we postulate that given an event in component $i$, it directly \emph{causes} $\mathrm{Pois}(a_{i,j})$ new events in component $j$. This makes the branching coefficient $a_{i,j}$ an obvious measure for the strength of the causal effect from component $i$ on component $j$. Such causal effects are often represented as directed graphs. In the literature on causality, a graphical approach for modeling the interdependence of event streams  can for instance be found in \cite{meek14} or \cite{gunawardana14}---without any mentioning of `Hawkes'. This shows how natural the definition of a Hawkes graph is. First, we introduce some general graph terminology:

\begin{definition}\label{def:graph}
Let $d\in\N$ and $[d] = \{1,2,\dots, d\}$. A \emph{graph} $\mathcal{G}$ is a 2-tuple $(\mathcal{V},\mathcal{E})$, where $\mathcal{V} =[d]$ is a set of  \emph{vertices} and $\mathcal{E}\subset \mathcal{V} \times \mathcal{V}$ is a set of \emph{edges}. Given such a graph $\mathcal{G}$ we introduce the following definitions:
\begin{enumerate}[label = \roman*)]
\item Vertex $i$ is a \emph{parent} of vertex $j$ if $(i,j)\in\mathcal{E}$. We write
$\PA(j):=\{i:\ (i,j)\in\mathcal{E}\}$. Vertex $i$ is a \emph{source vertex} if $\PA(i)\setminus\{i\}=\emptyset$.Vertex $i$ is a \emph{sink vertex} if $\{j:\, (i,j)\in\mathcal{E}\} \setminus\{i\}=\emptyset$.
\item For $g\in\N$, $(k_0, k_1,\dots, k_g)\in\mathcal{V}^{g + 1}$ is a \emph{walk} in $ \mathcal{G}$ of length $g$ from vertex $i$ to vertex $j$ if $k_0=i, k_g =j$ and $(k_{l-1},k_{l})\in\mathcal{E},\,l\in[g]$; $(k_0, k_1,\dots, k_g)\in\mathcal{V}^{g + 1}$ is a \emph{closed walk} if it is a walk with $k_0 = k_g$. 
We denote the \emph{set of walks in $\mathcal{G}$ from $i$ to $j$ with length $g\in\N$} by $\mathcal{W}_{g}^{(i,j)}$. 
Furthermore, we set $\mathcal{W}^{(i,j)}_0:=\emptyset$ if $i\neq j$, $\mathcal{W}^{(i,j)}_0:= \{(i)\}$ if $i = j$, $\mathcal{W}^{(i,j)}:=\cup_{g \geq 0}\mathcal{W}_g^{(i,j)}$, and $\mathcal{W}:=\cup_{(i,j) \in [d]^2}\mathcal{W}^{(i,j)}$.
\item Vertex $i$ is an \emph{ancestor} of $j$ if there exists a walk of length $g\in \N$ from $i$ to $j$. We denote the \emph{ancestor set} of a vertex $i$ in $\mathcal{G}$ by $\AN(i)$. 

\item The vertices $i$ and $j$
are \emph{weakly connected} if $i = j$ or if there exists a set
 $\{(k_{l-1},k_{l}),\, l = 1,\dots,g:\ k_0=i, k_g =j, (k_{l},k_{l-1}) \in\mathcal{E}\ \text{ or }\,(k_{l-1},k_{l}) \in\mathcal{E}  \}$ for some $g\in\N$. The vertices $i$ and $j$ are \emph{strongly connected} if the sets $\mathcal{W}^{i,j}$ and $\mathcal{W}^{j,i}$ are nonempty. A graph is \emph{weakly (strongly) connected} if all pairs of its vertices are weakly (strongly) connected. A graph is \emph{fully connected} if $(i,j)\in\mathcal{E},\ (i,j)\in [d]^2$. 
 \end{enumerate}
\end{definition}
Note that in our definition, a graph allows cycles and, in particular, self-loops. A vertex may or may not be an ancestor and, in particular, a parent of itself. Also note that any vertex $i$ is always strongly connected to itself because $\{(i)\}\subset\mathcal{W}^{(i,i)},\, i\in [d]$---no matter if $i$ is contained in a closed walk or not. Consequently, the singleton graph is always strongly connected. However, it is only fully connected if is a self-loop. Next, we apply  the graph terminology from Definition~\ref{def:graph} to the Hawkes process:
\begin{definition}\label{hg}
Let ${\bf N}$ be a $d$-type Hawkes process with immigration intensities $\eta_1,\eta_2,\dots,\eta_d$ and branching coefficients $a_{i,j} (= \int h_{i,j}(t)\d t )$, $(i,j)\in[d]^2$; see Definitions~\ref{hawkes_family} and \ref{def:hawkes_process}. The \emph{Hawkes graph skeleton} $\mathcal{G}^*_{\bf N} = (\mathcal{V}^*_{\bf N}, \mathcal{E}^*_{\bf N})$ of  ${\bf N}$ consists of a set of vertices $\mathcal{V}_{\bfN}^* = [d]$ and a set of edges
$$
\mathcal{E}^*_{\bf N}:=\Big\{(i,j)\in\mathcal{V}_{\bfN}^*\times\mathcal{V}_{\bfN}^* :\ 
a_{i,j}> 0 \Big\}.
$$
For $j\in[d]$, we denote the \emph{parent}, respectively, \emph{ancestor set} of $j$ with respect to the Hawkes skeleton $\mathcal{G}^*_{\bf N}$ by $\PA_{\bfN}(j)$ and $\AN_{\bfN}(j)$.
For the \emph{Hawkes graph} $\mathcal{G}_{\bf N} = (\mathcal{V}_{\bf N}, \mathcal{E}_{\bf N})$ of ${\bf N}$, each vertex, respectively, edge of the corresponding Hawkes skeleton is supplied with a \emph{vertex}, respectively, an \emph{edge weight}:
\begin{eqnarray*}
\mathcal{V}_{\bf N} &:=& \Big\{(j;\eta_j):\quad j\in \mathcal{V}^*_{\bf N}\ \mathrm{and}\ \eta_j\ \mathrm{is}\ \mathrm{the}\ j\text{-th}\ \mathrm{immigration}\ \mathrm{intensity}\ \mathrm{of}\ \bfN\Big\},\\
 \mathcal{E}_{\bf N} &:= &
\bigg\{(i,j; a_{i,j}):\quad (i,j)\in\mathcal{E}^*_{\bf N} \ \text{and}\ (a_{i,j})_{(i,j)\in[d]^2} \text{ is the branching matrix of }\bfN\bigg\}.
\end{eqnarray*}
We call the branching matrix $A = \(a_{i,j}\)\in\R_{\geq 0}^{d\times d}$ of $\bfN$ the \emph{adjacency matrix} of $\mathcal{G}_{\bfN}$.
\begin{enumerate}[label = \roman*)]
\item A Hawkes graph $\mathcal{G}_{\bfN}$ is \emph{weakly, strongly, respectively, fully connected} if the corresponding skeleton $\mathcal{G}^*_{\bfN}$ is weakly, strongly, respectively, fully connected; see Definition~\ref{def:graph}.
\item Vertex $(j;\eta_j)$ of a Hawkes graph $\mathcal{G}_{\bfN}$ is a \emph{source}, respectively, \emph{sink vertex}, if it is a source, respectively, sink vertex in the corresponding skeleton $\mathcal{G}^*_{\bfN}$.  Furthermore, $(j;\eta_j)$ is a \emph{redundant vertex} if $\eta_j = 0$ and, in addition, $\eta_i = 0$ for all $i\in \AN_{\bfN}(j)$.
\item
For any walk $w\in \mathcal{W}_{\mathcal{G}_{\bfN}}\, (:=\mathcal{W}_{\mathcal{G}_{\bfN}^*} )$ in a Hawkes graph $\mathcal{G}_{\bfN}$, we define the \emph{walk weights}
$$
|w| = |(i_0, i_1,\dots,i_g)| := 
\begin{cases}
1,& g = 0, \text{ and}\\
\prod_{l = 1}^g a_{i_{l-1},i_l}, & g > 0,
\end{cases}
$$
where $a_{i_{l-1}, i_l},\ l= 1,2,\dots, g,$ are the edge weights from $\mathcal{E}_{\bfN}$.
\item \label{item:iv} A Hawkes graph is \emph{subcritical} if 
\begin{align}
\sum_{w \in {\mathcal{W}}^{(i_0, i_0)}} |w|& <\infty,\, i_0\in[d]\text{, or, equivalently, }
\sum_{\substack{w:\\ w \text{ closed walk in } \mathcal{G}_{\bfN}}} |w| <\infty.
\label{subcriticality}
\end{align}
\end{enumerate}
\end{definition}
Note that if a Hawkes graph vertex is redundant, then all its ancestors are also redundant. The notion of a subcritical Hawkes graph in Definition~\ref{hg}~\ref{item:iv} might ask for further explanation. The following theorem clarifies things:
\begin{theorem}\label{prop:graph_subcriticality}
 Let $\bfN$ be a Hawkes process and let $\mathcal{G}_{\bfN}$ be the corresponding Hawkes graph. Then $\bfN$ is a subcritical Hawkes process (see Definition~\ref{def:hawkes_process}) if and only if $\mathcal{G}_{\bfN}$ is a subcritical  Hawkes graph (see Definition~\ref{hg}).
\end{theorem}
\begin{proof}
First, we prove that
\begin{align}
\sum\limits_{w\in{\mathcal{W}}^{(i_0, i_0)}}
|w| < \infty,\, i_0 \in[d] \quad
\Leftrightarrow 
\quad
\sum\limits_{w\in{\mathcal{W}}^{(i_0, j)}}
|w| < \infty,\, (i_0,j)\in[d]^2.\label{eq:equivalence}
\end{align}
`$\Leftarrow$' is trivial. We show `$\Rightarrow$' by induction over the graph size $d$: for $d = 1$, the implication is true. For $d > 1$, consider a graph with $d$ vertices and assume that the left-hand side of \eqref{eq:equivalence} holds. 
Pick any $(i_0, j)\in[d]^2$.
 We split the possible paths from $i_0$ to $j$, $\mathcal{W}^{(i_0, j)}$, into \emph{paths excluding $d$}, $\mathcal{W}^{(i_0, j)}_{\text{excl.} d}$, and \emph{paths including $d$}, $\mathcal{W}^{(i_0, j)}_{\text{incl.} d}$: 
\begin{align}
\sum\limits_{w \in{\mathcal{W}}^{(i_0, j)}}|w| 
 = \sum\limits_{w \in\mathcal{W}^{(i_0, j)}_{\text{excl.} d}} |w|
+  \sum\limits_{w \in\mathcal{W}^{(i_0, j)}_{\text{incl.} d}} |w|.\label{eq:split1}
\end{align}
The first sum is finite by the induction hypothesis. Now, assume the case that $i_0 \neq d$ and $j \neq d$. Every walk in the second sum of \eqref{eq:split1} may be (uniquely) split into the following five subwalks: a $d$-avoiding walk $w_1$ from $i_0$ to some $i_1\in\PA(d)$, a one-step walk $(i_1, d)$, a walk $w_2$ $\in \mathcal{W}^{(d,d)}$, another one-step walk $(d, j_1)$, with $d\in \PA(j_1)$, and finally some $d$-avoiding walk $w_3$ from $j_1$ to $j$. This yields
\begin{align*}
&\hspace{-0.5cm}\sum\limits_{w \in \mathcal{W}_{\text{incl.} d}^{(i_0, j)}} |w| \\
&= \sum_{i_1\in \PA(d)}\sum_{w_{1}\in{\mathcal{W}}^{(i_0, i_1)}_{\text{excl.} d}}\sum_{w_2 \in{\mathcal{W}}^{(d, d)}}\sum_{j_1:  d\in \PA(j_1)}\sum_{w_{3} \in {\mathcal{W}}^{(j_1, j)}_{\text{excl.} d}}|w_{1}|\, a_{i_1, d} \,|w_2| \, a_{d,j_1}\, |w_3 |\\
&\leq \sum\limits_{i_1\in \PA(d)}\sum_{j_1:  d\in \PA(j_1)} \max_{(i,j)\in [d]^2} a_{i,j}^2 \underbrace{\sum_{w_{1}\in{\mathcal{W}}^{(i_0, i_1)}_{\text{excl.} d}}|w_{1}|}_{< \infty\text{ by ind. hyp.}} \underbrace{\sum_{w_{2} \in {\mathcal{W}}^{(d, d)}} |w_{2}|}_{< \infty\text{ by assumption}}\underbrace{\sum_{w_{3} \in {\mathcal{W}}^{(j_1, j)}_{\text{excl.} d}} |w_{3} |}_{< \infty\text{ by ind. hyp.}}< \infty.
\end{align*} 
Note that, by definition, $(i)\in \mathcal{W}^{(i,i)}$ and $|(i)| = 1,\, i\in [d],$ so that the calculation above also covers the cases $ \PA(d) = \{i_0\} $ and $\PA(j) = \{d\}$ as well as $d$-including walks from $i$ to $j$ that touch $d$ exactly once. If $i_0 = d$ or $j = d$, the splitting argument becomes even simpler; we do not give the details. We have proven the finiteness of the second sum in \eqref{eq:split1} and therefore \eqref{eq:equivalence}.\\
\par Next, note that
\begin{align}
\sum\limits_{w\in{\mathcal{W}}^{(i_0, j)}}
|w| = \sum\limits_{g\geq 0}\sum\limits_{w\in\mathcal{W}_g^{(i_0, j)}}
|w|=\sum\limits_{g\geq 0} \E Y^{(i_0)}_{g,j}= \E F^{(i_0)}_j(\R),\quad (i_0,j) \in[d]^2, \label{eq:expectation}
\end{align}
where $(\bfY^{(i_0)}_{g}) = (Y^{(i_0)}_{g,1},Y^{(i_0)}_{g,2}, \dots , Y^{(i_0)}_{g,d})$ are the embedded generation processes of the generic family processes $\bfF^{(i_0)}= (F^{(i_0)}_1, \dots, F^{(i_0)}_d)$ of $\bfN$; see Definition~\ref{embedded_generation_process}. Thus, \eqref{subcriticality} is a complicated way of saying that, for all $(i_0,j)\in[d]^2$, the expected total number of type-$j$ offspring events of a type-$i_0$ event is finite, i.e., that $\E F^{(i_0)}_{j}(\R) < \infty,\, (i_0,j)\in[d]^2$. By Proposition \ref{prop:expectation}, this in turn is equivalent to the spectral radius of the branching matrix being strictly less than 1---which is the original Hawkes-\emph{process} subcriticality condition from Definition~\ref{def:hawkes_process}
\end{proof}
 Obviously, the Hawkes graph does not fully specify the corresponding Hawkes process; it only captures the structure of the embedded generation processes from Definition~\ref{embedded_generation_process} together with the immigration intensities. Despite this simplification, the Hawkes graph gives relevant insight into the underlying Hawkes process---especially in the highdimensional case. For example, \emph{connectivity and redundancy of vertices} are two graph-based concepts that become increasingly important the higher the dimension of the model considered is. If a Hawkes graph is not weakly connected, we may consider the \emph{weakly connected} subgraphs separately and correspondingly split the original model into separate, lower-dimensional Hawkes processes. The notion of {\it redundant vertices} is important because, typically, we only want to consider `accessible' event types. {\it Sink} ({\it source}) \emph{vertices} of a Hawkes graph correspond to Hawkes process components that only receive (give) excitement from (to) the system. The notion of \emph{parent sets} is also helpful: e.g., for the marginal conditional intensity in \eqref{multivariate_intensity}, it is actually enough to sum over $i\in\PA(j)$ instead of $i\in[d]$ which may be computationally beneficial. The \emph{ancestor sets} may be applied if we are only interested in modeling events of a particular type $j$. In this situation, it suffices to consider a Hawkes model for the event types in $\{j\}\cup\AN(j)$.
Finally, we find the formulation of Hawkes graph \emph{subcriticality} in \eqref{subcriticality} useful. It provides a more concrete meaning to the somewhat abstract eigenvalue-based criterion for the Hawkes process. E.g., \eqref{subcriticality} can be used when constructing subcritical Hawkes graphs, respectively, models. And---if a given graph is sparse and the closed walks are not too numerous---one can check subcriticality without even calculating any eigenvalue; see Section~\ref{example_model}. Furthermore, in some cases, the \emph{path weights} $|w|$ themselves might be worth calculating---even apart from criticality conditions; see the discussion in the proof of Theorem~\ref{prop:graph_subcriticality}. Last but not least, the graph structure obviously allows for attractive self-explaining illustrations; see Figures~\ref{fig1} and~\ref{fig2}. In the following proposition, we collect some specific graphical and statistical information that may be calculated from the adjacency matrix of a Hawkes graph:

\begin{proposition}\label{adj}
For some $d\geq 2$, let $\bfN$ be a $d$-type subcritical Hawkes process. Furthermore, let $\mathcal{G}_{\bfN} = (\mathcal{V}_{\bfN}, \mathcal{E}_{\bfN})$ be the corresponding Hawkes graph with adjacancy matrix $A = (a_{i,j})\in\R^{d\times d}_{\geq 0}$. Then we have that
\begin{enumerate}[label = \roman*)]
\item $a_{i,j}>0\quad \Leftrightarrow\quad i\in\PA_{\bfN}(j)$;
\item $a_{i,j}=0,\, j\in[d]\setminus\{i\}$ $\quad \Leftrightarrow\quad$ vertex $i$ is a sink vertex;
\item $a_{i,j}=0,\, i\in[d]\setminus\{j\}$ $\quad \Leftrightarrow\quad$ vertex $j$ is a source vertex;
\item $(A^g)_{i,j}>0\quad \Leftrightarrow \quad$ there is a walk of length $g$ from $i$ to $j$;
\item $(A^g)_{i,j}>0\ \text{ for some }g\in[d]\quad \Leftrightarrow$\quad$i\in\AN(j)$;
\item for all $(i,j)\in[d]^2$, $((A + A^{\top})^g)_{i,j}>0\ \text{ for some }g\in\{0\}\cup[d - 1]
\quad \Leftrightarrow \quad$ the Hawkes graph $\mathcal{G}_{\bfN}$ is weakly connected;
\item for all $(i,j)\in[d]^2$, $((A)^g)_{i,j}>0\ \text{ for some }g\in\{0\}\cup[d - 1]
\quad \Leftrightarrow \quad$ the Hawkes graph $\mathcal{G}_{\bfN}$ is strongly connected;
\item $a_{i,j}>0,\, (i,j)\in[d]^2$  $\quad \Leftrightarrow\quad$ the Hawkes graph $\mathcal{G}_{\bfN}$ is fully connected;
\end{enumerate}
\end{proposition}

The properties above can easily be checked. They may help to describe the relationships between Hawkes process components, respectively, Hawkes graph vertices. Two specific $\mathbb{R}_{\geq 0}^d$-vectors might be particularly meaningful statistical summaries of a Hawkes graph, respectively, Hawkes process:

\begin{definition}\label{coefficients}
Let ${\bf N}$ be a subcritical $d$-type Hawkes process and let $A$ be the adjacency matrix of the corresponding Hawkes graph $\mathcal{G}_{\bfN}$. Consider 
the limit matrix $\mathbb{R}_{\geq 0}^{d\times d}\ni (e_{i,j}) := (1_{d\times d} - A)^{-1} = \sum_{g\geq 0} A^g\, (= (\E F^{(i_0)}_j(\R))_{(i_0, j)\in[d]^2})$ from Proposition~\ref{prop:expectation} and define
$$
c_{i_0} :=  \frac{\eta_{i_0}\sum_{j = 1}^d e_{i_0,j}}{\sum_{i=1}^d \eta_{i}\sum_{j= 1}^d e_{i,j }},\quad {i_0}\in [d], \quad\text{and}\quad
 f_j := \frac{\eta_j e_{j,j}}{{\sum_{i = 1}^d\eta_i{e_{i,j}}}},\quad j\in [d].
 $$
 We call $(c_{i_0})_{i \in [d]}$ the \emph{cascade coefficients} and $(f_j)_{j \in[d]}$ the \emph{feedback coefficients.}
 \end{definition}
\par 
One way of tuning a specific Hawkes graph may be achieved by `switching-off' a selected vertex by forcing the corresponding immigration intensity to zero. The coefficients defined above summarize the effect of such a manipulation. In view of Proposition~\ref{prop:expectation}, we have the following interpretations. First of all, the \emph{cascade coefficients} $(c_i)$ are important from a \emph{systemic} point of view. The cascade coefficient $c_i$ measures the  fraction of events in the system stemming from families with immigrated type-$i$ ancestor. If $c_i > 1/d$, this indicates a relatively large impact of type-$i$ events on the system. 
Secondly, the \emph{feedback coefficients} $(f_j)$ are more important from an \emph{individual} point of view. They indicate how much of the total intensity that a vertex $j$ \emph{experiences} is due to its own immigration activity including the feedback it experiences by closed walks. We illustrate both concepts in Section~\ref{example_model}.
\section{Estimation}\label{estimation}
In this section, we give a summary of earlier work, where we introduced a nonparametric estimation procedure for the multivariate Hawkes process. Based on this approach, we introduce an estimation procedure for the Hawkes skeleton and the Hawkes graph. In particular, we clarify how one can bypass numerical problems in high-dimensional settings. Finally, we explain how one can use the results for completely specifying and estimating a parametric Hawkes model. 

\subsection{Earlier results}\label{earlier_results}
In \citep{kirchner16c}, we showed that the distributions of the bin-count sequences of a Hawkes process can be approximated by the distribution of so called \emph{integer-valued autoregressive time series} INAR(p). This approximation yields an estimation method for the Hawkes process: we fit the approximating model on observed bin-counts of point process data. The resulting estimates can be used as estimates of the Hawkes reproduction intensities on a finite and equidistant grid; see \cite{kirchner16d}. For illustration, consider a univariate Hawkes process $N$ with reproduction intensity $h$ and immigration intensity $\eta$. Given data from $N$ in a time window $(0,T]$, $\Delta>0$, small, bin counts $X^{(\Delta)}_n := N\big(((n-1)\Delta, n\Delta]\big),\, k = 1,2,\dots,n :=\lfloor T/\Delta\rfloor$, and some $p\in\N$, large, we calculate 
\begin{align}
\(\hat{\alpha}^{(\Delta)}_0,\hat{\alpha}^{(\Delta)}_1,\dots, \hat{\alpha}^{(\Delta)}_p\) := \mathrm{argmin}_{(\alpha_0^{(\Delta)},\alpha_1^{(\Delta)},\dots,\alpha_p^{(\Delta)})}\sum\limits_{k=p+1}^{n}
\(X^{(\Delta)}_k - \alpha^{(\Delta)}_0-\sum\limits_{l=1}^p\alpha^{(\Delta)}_lX^{(\Delta)}_{k-l}\)^2. \label{errors}
\end{align}
Given~\eqref{errors}, we
estimate the reproduction-intensity values $h(k\Delta)$, $k=1,2,\dots,p$, of $N$ by $\hat{h}_k := \hat{\alpha}_k^{(\Delta)}/{\Delta}$ and the immigration intensity $\eta$ by $\hat{\eta} := \hat{\alpha}^{(\Delta)}_0/\Delta$.
The multivariate case is conceptually equivalent but somewhat cumbersome notationwise. 
Furthermore---due to the special distribution of the errors---the covariance matrix of the estimates is nonstandard. This is why we give all formulas in some detail. The following definitions and properties are taken from \cite{kirchner16d}---modulo transposition as stated in Remark~\ref{ij}. 

\begin{definition} \label{estimator}
Let $\bfN= \left(N_1,N_2,\dots,N_d\right)$ be a subcritical $d$-type Hawkes process with immigration intensity $\bfeta\in\R_{\geq 0}^d\setminus \{0_d\}$ and reproduction intensities $h_{i,j}: \R_{\geq 0}\to \R_{\geq 0}$, $(i,j)\in [d]^2 $. 
Let $T>0$ and
consider a sample of the process on the time interval $(0,T]$.
For some $\Delta>0$, construct the $\N^d_0$-valued \emph{bin-count sequence} from this sample:

\begin{equation}
\bfX^{(\Delta)}_k := \bfN \Big(\big((k-1)\Delta,k\Delta\big]\Big)^\top\in\N_0^{d\times 1},\quad k = 1, 2, \dots , n := \left\lfloor T/\Delta\right\rfloor.\label{bin_count_sequence}
\end{equation}
Define the \emph{multivariate Hawkes estimator} with respect to some support $s,\, \Delta < s <T$,
\begin{equation}
\widehat{\bf{H}}^{{(\Delta,s)}}:=  \frac{1}{\Delta}\left(\bfZ^\top \bfZ \right)^{-1} \bfZ^\top\bfY\quad \in\R^{ (dp+1)\times d}.\label{estimator_calculation}
\end{equation}Here,

\begin{align}
\bfZ\left(\mathbf{X}^{(\Delta)}_1,\dots ,\mathbf{X}^{(\Delta)}_n\right):=\left(\begin{array}{ccccc}
(\mathbf{X}^{(\Delta)}_{p})^\top & (\mathbf{X}^{(\Delta)}_{p-1})^\top& \dots   &(\mathbf{X}^{(\Delta)}_{1})^\top&1 \\
(\mathbf{X}^{(\Delta)}_{p + 1})^\top & (\mathbf{X}^{(\Delta)}_{p})^\top& \dots  & (\mathbf{X}^{(\Delta)}_{2})^\top&1\\
\dots  & \dots  & \dots  & \dots  &\dots\\
(\mathbf{X}^{(\Delta)}_{n-1})^\top& (\mathbf{X}^{(\Delta)}_{n-2})^\top & \dots    & (\mathbf{X}^{(\Delta)}_{n-p})^\top & 1
\end{array}\right)\in\R^{ (n-p)\times (dp+1)}\label{design_matrix}
\end{align}
is the \emph{design matrix}
and $\bfY\left(\mathbf{X}^{(\Delta)}_1,\dots ,\mathbf{X}^{(\Delta)}_n\right):=\left(\mathbf{X}^{(\Delta)}_{p+1}, \mathbf{X}^{(\Delta)}_{p+2}, \dots  ,\mathbf{X}^{(\Delta)}_{n}\right)^\top\in\R^{(n-p)\times d}
$ with $p:= \lceil s / \Delta\rceil$.
\end{definition}
For the following considerations, we drop the `${(\Delta,s)}$' superscript.  Note that also the matrices $\bfZ$ and $\bfY$ depend on $\Delta$. Additional notation clarifies what the entries of the matrix  $\widehat{\bf{H}}$ in \eqref{estimator_calculation} actually estimate:
\begin{align}\label{estimator_interpretation}
\left(\begin{array}{c}
\widehat{H}_1\\
 \dots \\
  \widehat{H}_p\\
  \hat{\bfeta}
\end{array}\right)
:= \widehat{\bf{H}} \in\R^{ (dp + 1)\times d},\quad \text{where}\quad \widehat{H}_k :=  \left(\begin{array}{cccc}
\hat{h}_{1,1}(k\Delta)& \hat{h}_{1,2}(k \Delta)&\dots & \hat{h}_{1,d}(k \Delta)\\
 \hat{h}_{2,1}(k\Delta)& \hat{h}_{2,2}(k \Delta)&\dots & \hat{h}_{2,d}(k \Delta)\\
\dots& \dots&\dots & \dots\\
\hat{h}_{d,1}(k\Delta)& \hat{h}_{d,2}(k \Delta)&\dots & \hat{h}_{d,d}(k \Delta)\\
\end{array}\right).
\end{align}
In \citet{kirchner16d}, we find that, for large $T$, small $\Delta$ and large $p$, the entries of $\widehat{\mathbf{H}}$ are approximately jointly normally distributed around the true values. Furthermore, the covariance matrix of $\mathrm{vec}\left(\widehat{\mathbf{H}}^\top\right)\in\R^{d(dp + 1)}$ ($\mathrm{vec}(\cdot)$ stacks the columns of its argument) can be consistently estimated by

\begin{equation}
\widehat{S^2}:=\frac{1}{\Delta^2}\left(\left(\bfZ^\top\bfZ\right)^{-1}  \otimes 1_{d\times d}\right) 
{\bf W}
\left(\left(\bfZ^\top\bfZ\right)^{-1}  \otimes 1_{d\times d}\right)\in \R^{d(dp+1) \times d(dp+1)}.\label{cov_estimator}
\end{equation}
Here, $\otimes$ denotes the Kronecker product, $\bfZ$ is the design matrix from \eqref{design_matrix} and ${\bf W}:= \sum_{k={p+1}}^n{\bfw_k \bfw_k^\top}\in \R^{d(dp+1) \times d(dp+1)}$, where, for $k= p+1, p+2, \dots, n$,
\begin{eqnarray}
\bfw_k
&:=& \left(\left(\left(\bfX^{(\Delta)}_{k-1}\right)^\top, \left(\bfX^{(\Delta)}_{k-2}\right)^\top,\dots ,\left(\bfX^{(\Delta)}_{k-p}\right)^\top,1\right)^\top\otimes 1_{d\times d}\right)\\
&&\hspace{4cm}\cdot\,\left(\bfX^{(\Delta)}_k - \Delta\hat{\bfeta} - \sum\limits_{l=1}^p\Delta \widehat{H}_l^\top \bfX^{(\Delta)}_{k-l} \right)\in\R^{d(dp  + 1)\times 1}.\nonumber
\end{eqnarray}

In Definition~\ref{estimator}, we consider  $\vec(\bf{H}^\top)$ instead of $\vec(\bf{H})$ in order to apply the results from \citet{kirchner16c} more directly; see Remark~\ref{ij}. We will discuss below how one retrieves specific values from the covariance matrix estimation in~\eqref{cov_estimator}. 
The estimator from Definition~\ref{estimator} above depends on a support $s,\,0<s<<T,$ and on a bin size $\Delta,\,0<\Delta\leq s$. Automatic methods for the choice of these estimation parameters are discussed in \citet{kirchner16c}. In the present paper, we assume $s$ given. Often, an upper bound for the support of the reproduction intensities can be guessed from the data context. The choice of $\Delta$, however, will be crucial in high-dimensional settings. We will use it as a tuning parameter for controlling numerical complexity.

\subsection{Estimation of the Hawkes skeleton}\label{estimation_of_the_hawkes_skeleton}
Our first goal is to identify the edges of the Hawkes skeleton from data; see Definition~\ref{hg}. The idea is simple: for $(i,j) \in [d]^2$, we estimate the edge weight $a_{i,j} = \int h_{i,j}(t)\d t$ by
$\hat{a}_{i,j}  := \Delta \sum_{k = 1}^p \hat{h}_{i,j} (k\Delta)$; see \eqref{estimator_interpretation} for the notation. Calculating the covariance estimate \eqref{cov_estimator}, we can check whether $\hat{a}_{i,j}$ is significantly larger than zero. If this is the case, we set $(i,j)\in\widehat{\mathcal{E}}^*$.
In order to ease implementation, we explicitly give the necessary transformations for the estimates from Definition~\ref{estimator} and discuss numerical issues. 

\begin{definition}\label{skeleton_estimator}
Given $d$-type event-stream data on $(0,T]$, calculate the Hawkes estimator $\bfH^{(\Delta_{\text{skel}}, s)}$ from Definition~\ref{estimator} with respect to some $s,\, 0< s <T,$ and some $\Delta_{\text{skel}},\,0 < \Delta_{\text{skel}} \leq s$. For $j\in[d]$, let  $b_j\in\{0,1\}^{(dp + 1)\times 1}$ be column vectors with all entries zero but 1s at entries $(k - 1)d + j,\, k=1,2,\dots, p = \lceil s / \Delta_{\text{skel}} \rceil$. Let $B:= (b_1,b_2,\dots,b_d)^\top$, and calculate
\begin{align}
\(\hat{a}_{i,j}\)_{1\leq i,j\leq d} = \Delta_{\text{skel}} B\bfH^{(\Delta_{\text{skel}}, s)}.\label{a_ij}
\end{align}
Fix $\alpha_{\text{skel}}\in(0,1)$ and define the
 \emph{Hawkes-skeleton estimator} as a graph $\widehat{\mathcal{G}}^* := ([d],\widehat{\mathcal{E}}^* ),$ with
\begin{align}
  \widehat{\mathcal{E}}^*:= \Big\{(i,j) \in [d]^2:\ \hat{a}_{i,j} > \hat{\sigma}_{i,j} z^{-1}_{1 - \alpha_{\text{skel}}} \Big\}.\label{edge_estimator}
\end{align}
Here, for $\beta\in(0,1)$, $z^{-1}_{\beta}$ denotes the $\beta$-quantile of a standard normal distribution. Efficient calculation of $(\hat{\sigma}_{i,j})_{1\leq i,j\leq d}$ will be given in Algorithm~\ref{sigma_comp_alt} below.
\end{definition}
The main point of this first estimation step is that we hope that the edge set $|\mathcal{E}^*|$ and, consequently  $|\widehat{\mathcal{E}}^*|$ are typically much smaller than $d^2$, respectively, that
$
\PA_{{\bfN}}(j),\, j\in[d],
$
and, consequently, 
$
\widehat{\PA}_{\bfN}(j),\, j\in[d],
$
are typically much smaller than $d$. If this is the case, the knowledge of the skeleton simplifies the estimation of the Hawkes graph considerably.
\subsubsection*{The role of $\Delta_{\text{skel}}$} On the one hand, the smaller we choose the bin size $\Delta$, the better the discrete approximation described in Section~\ref{earlier_results} works. On the other hand, the matrices involved in the calculation of the Hawkes estimator from Definition~\ref{estimator} become increasingly large when $\Delta$ decreases. More specifically, \eqref{estimator_calculation} involves the construction and multiplication of matrices with about $ ds/\Delta$ rows and  about $T/\Delta$ columns, where $T>0$ denotes the sample window size, $d\in\N$ the number of event-types, and $s, \, \Delta \leq s <<T$, the support parameter from Definition~\ref{estimator}. Furthermore, we have to invert matrices of size $ \lceil ds/\Delta\rceil \times \lceil ds/\Delta\rceil$. The crucial observation is that in the Hawkes-skeleton estimation, we may choose $\Delta_{\text{skel}}$ quite large for two reasons:
\begin{enumerate}[label = \roman*)]
\item The test involved in \eqref{edge_estimator} does not depend on $\Delta_{\text{skel}}$ too heavily. The false positive rate (that is, the probability of \emph{including a false edge}) is well controlled by $\alpha_{\text{skel}}$, because, under $H_0: h_{i,j}\equiv 0$, discretizations as in~\eqref{errors} stay `correct' even for very coarse $\Delta_{\text{skel}}$; see~\eqref{eq:prob} below.
The false negative rate (probability of \emph{missing a true edge}) naturally depends strongly on the true underlying edge weights. However, if there is truly considerable direct excitement from one component to another, then typically the effect from some bin to future bins will also be of some significance---which is exactly what our skeleton estimator tests. Our simulation study in Section~\ref{simulation_study} confirms these arguments. 
\item  The actual \emph{quantitative} estimation of the interactions between different event types will be performed in a second step when we consider the Hawkes \emph{graph}. In this second step, due to the (hoped-for) sparseness of the Hawkes skeleton, we are typically able to choose a much finer bin size $\Delta_{\text{graph}}$. So we may ignore the bias stemming from a somewhat rough discretization in the first (skeleton-estimation) step. 
\end{enumerate}
By choosing $\Delta_{\text{skel}} = s/k$ for some small $k\in\N$ in the calculations of Definition~\ref{skeleton_estimator} above, even Hawkes-skeleton estimates of very high-dimensional models (such as $d>20$) become computationally tractable. 

\subsubsection*{The role of $\alpha_{\text{skel}}$} 
Note that under
 $H_0:\, a_{i,j} \equiv 0$, we have that
\begin{align}\label{eq:prob}
{\P}_{H_0}[\hat{a}_{i,j} > \hat{\sigma}_{i,j}^2 z^{-1}_{1 - \alpha_{\text{skel}}} ] \approx \alpha_{\text{skel}}.
\end{align}
Still, the parameter $\alpha_{\text{skel}}\in (0,1)$ should not so much be thought of as an actual significance level---due to the multiple testing setup over $(i,j)\in[d]^2$, and because of the dependence between the different edge tests. Despite this warning, note that in the simulation study from Section~\ref{simulation_study}, 
the corresponding empirical false positive rates are very close to our (varying) choices of $\alpha_{\text{skel}}$. In any case, $\alpha_{\text{skel}}$ is a flexible tuning parameter that allows for controlling the degree of sparseness in the estimated graph. A value of $\alpha_{\text{skel}} = 1$ will yield a fully connected estimated graph as Hawkes skeleton. When $\alpha_{\text{skel}}$ decreases, the skeleton estimate becomes sparser and sparser. For $\alpha_{\text{skel}} \geq 0.01$, we typically  still \emph{overestimate} the true edge set. In other words, for $j\in[d]$, we typically have that $\PA_{\bfN}(j)\subset\widehat{\PA}_{\bfN}(j)$ with high probability. 

\subsubsection*{Variance estimate calculation}
The most elaborate step from a computational point of view in Definition~\ref{estimator} is the calculation of the covariance estimator in \eqref{cov_estimator}. Here, we deal with matrices of size $ \lceil d^2s/\Delta\rceil \times \lceil d^2s/\Delta\rceil$. Furthermore, we have to calculate approximately $T/\Delta$ vectors of size $d^2s/\Delta$ and calculate and sum their crossproducts $\bfw_k\bfw_{k}^\top$. This is the numerical bottleneck of the procedure---in particular for high-dimensional setups. For the Hawkes-skeleton estimator from Definition~\ref{skeleton_estimator}, we simplify the calculation. First of all, we note that in the matrix $\widehat{S}^2$ from \eqref{cov_estimator}, we estimate many more covariance values than we actually need for the (marginal) distribution of the edge-weight estimates. After some linear algebra, we find that one can avoid the tedious computation of the ${\bf W}$ matrix from (\ref{cov_estimator}) by the following matrix manipulations.

\begin{algorithm}\label{sigma_comp_alt}
Let $\bfE\in\{0,1\}^{d^2\times(d^2p +d)}$ be a matrix with all entries zero but, for $(i,j) = [d]^2$, in row
$(i - 1) d + j$ we have 1s at entries $(k-1)d^2 + (i-1)d
+j,\ k=1,2,\dots,p.$ Let $\bfE_{l,\cdot}$ denote the $l$-th row of $\bfE$.
With $\widehat{S}^2$ from \eqref{cov_estimator} and for $(i,j)\in[d]^2$, we have that
$
\hat{\sigma}_{i,j}^2:=\Delta^2 \bfE_{(i - 1) d + j,\cdot}^\top \widehat{S}^2 \bfE_{(i - 1) d + j,\cdot}
$
are the variance estimates for the $\hat{a}_{i,j}$ from $\eqref{a_ij}$.
These estimates can be computed in the following way:
\begin{enumerate}[label = \roman*)]
\item Compute $\bfE \big(\bfZ^\top \bfZ)^{-1} \bfZ^\top \otimes 1_{d\times d}\in\R^{d^2\times d(n-p)}$
and stack the rows of the result in a vector. Fill this vector row-wise in a $d^2(n-p)\times d$ matrix $\mathbf{C}$.
\item Set 
$
\bfU =  (\bfY - \Delta  \bfZ  \widehat{\bfH})\in\R^{(n-p)\times d}$. Denoting $ (U_{p+1}, U_{p + 2},\dots, U_n)^\top:= \bfU$, we now have that
$$
U_k = \left(\bfX^{(\Delta)}_k - \Delta\hat{\bfeta} - \sum\limits_{l=1}^p\Delta \widehat{H}_l^\top \bfX^{(\Delta)}_{k-l} \right),\quad k = {p+1}, p +2,\dots, n.
$$
Furthermore, let $\bfU^{\text{(rep)}}\in\R^{d^2(n-p)\times d}$ be a matrix consisting of $d^2$ repetitions of the $\bfU$ matrix stacked on top of each other.
\item  Multiply $\mathbf{C}$ from (i) pointwise with $\bfU^{\text{(rep)}}$ from (ii) and square the row sums of the resulting matrix. Row-wise fill the resulting vector into a $d^2\times (n-p)$ matrix and compute the row sums of this matrix.
\item Row-wise fill the result from (iii) into a $d\times d$ matrix. This yields $\(\hat{\sigma}_{i,j}^2\)_{1 \leq i,j\leq d}.$
\end{enumerate}
\end{algorithm}

\subsection{Estimation of the Hawkes graph}\label{estimation_of_the_hawkes_graph}
Given an estimate $\widehat{\mathcal{G}}_{\bfN}^*$ of the Hawkes skeleton $\mathcal{G}_{\bfN}^*$ from Definition~\ref{skeleton_estimator}, we consider the estimation of the Hawkes graph $\mathcal{G}_{\bfN}$; see Definition~\ref{hg}. We aim to estimate vertex as well as edge weights, and to calculate corresponding confidence bounds for both. That is, after the more structural Hawkes-skeleton estimation from Section~\ref{estimation_of_the_hawkes_skeleton}, we now \emph{quantify} the various interactions between the observed event streams. Typically, after the skeleton estimation, we can reduce the effective dimensionality of the model considerably: in a first obvious step, we divide the skeleton $\widehat{\mathcal{G}}_{\bfN}^*$ into its weakly-connected subgraphs and treat them separately. In a second step, we identify $\widehat{\PA}_{\bfN}(j):=\{i \in \mathcal{V}_{\bfN}:\, (i,j)\in\widehat{\mathcal{E}}_{\bfN}^*\}$ for all $j \in \mathcal{V}_{\bfN}$.
From the branching construction of a Hawkes process, respectively, of Hawkes families in Definitions~\ref{hawkes_family} and \ref{def:hawkes_process}, we have that any event in component $j$ is either an immigrant stemming from a Poisson random measure with constant intensity $\eta_j$ or has a direct explanation through an event in one of  its parent components $\PA_{\bfN}(j)$. That is, in a multivariate version of \eqref{errors}, \emph{it suffices to regress the bin-counts of component $j$ on the bin-counts in $\PA_{\bfN}(j)$}. The constant term in this regression will represent the $j$-th immigration intensity. Considering only the parents instead of all of the $d$ other components in the conditional-least-squares regression increases numerical efficiency and decreases estimation variance. 
In applications, however, we \emph{do not know} the true parent set $\PA_{\bfN}(j)$. So, we have to substitute ${\PA_{\bfN}}$ with the estimate $\widehat{\PA}_{\bfN}$. As long as $\PA_{\bfN}(j)\subset \widehat{\PA}_{\bfN}(j)$ this is not an issue: from the branching construction, we have that the intensity at time $t$ of component $j$,
conditional on $\sigma(N_i(A):\, A\in\B((-\infty,t]),\, i\in\PA_{\bfN}(j))$, is independent of the past of all other components $\sigma(N_i(A):\, A\in\B((-\infty,t]),\, i\notin\PA_{\bfN}(j))$. Consequently, additional vertices in the estimated parent sets do not introduce additional bias in this graph estimation.
Apart from this restriction of the regression variables on (estimated) parent types, we apply the conditional-least-squares approach as in Definition~\ref{estimator}. This time however, due to reduction of dimensionality, we will typically \emph{be able to choose a much smaller bin size} $\Delta_{\text{graph}}$ than for the skeleton estimation before. To ease implementation, below we give convenient notations and the necessary calculations.
\par
First, we drop the $\bfN$ subscript for the parent sets $\PA(j)$. Also, we write $\PA(j)$ instead of $\widehat{\PA}(j)$---keeping in mind that the first has to be substituted by the latter in most applications.
 For $k=1,2,\dots,n$, $j\in[d]$ and some $0 < \Delta_{\text{graph}}<<\Delta_{\text{skel}}$, let $\bfX^{(\Delta_{\text{graph}})}_{k,j}:= N_j\big(((k- 1)\Delta_{\text{graph}}, k\Delta_{\text{graph}}]\big)$, $d_j := |\PA(j)|$, and
\begin{align}
\bfX^{(\Delta_{\text{graph}})}_{k,\PA(j)} := \(\bfX^{(\Delta_{\text{graph}})}_{k,i_1}, \bfX^{(\Delta_{\text{graph}})}_{k,i_2}, \dots, \bfX^{(\Delta_{\text{graph}})}_{k,i_{d_j}}\)^\top. \label{notation}
\end{align}
 In \eqref{notation} and in what follows, we denote
$\{i_1,i_2,\dots,i_{d_j}\} := \PA(j)$ such that $i_1<i_2<\dots<i_{d_j}$. The idea is to regress all the bin counts of all $d$ event types separately on the past of their parents with Ansatz
\begin{align}
\E\[{\bf X}^{(\Delta_{\text{graph}})}_{n,j}\Big| {\bf X}^{(\Delta_{\text{graph}})}_{n-k,\PA(j)}, \, k= 1,2,\dots, p\] = \alpha^{(\Delta_{\text{graph}})}_{0,j} +  \sum\limits_{i \in \PA(j)}\sum\limits_{k = 1}^p\alpha^{(\Delta_{\text{graph}})}_{k,i,j}{\bf X}^{(\Delta_{\text{graph}})}_{n-k,i},\quad j\in[d].\label{mv_reg}
\end{align}
Ansatz \eqref{mv_reg} should be compared with \eqref{errors}.
Note that $j$ itself may or may not be an element of $\PA(j)$. 
\begin{definition}\label{def:graph_estimator}
Let $\mathcal{G}^*_{\bfN}$ be a Hawkes skeleton (estimate) with respect to some $d$-type Hawkes process (data) $\bfN$. Given $d_j := |\PA(j)|,\, j\in [d]$, a bin size $\Delta_{\text{graph}}>0$, a  support $s$ with $0 < \Delta_{\text{graph}} \leq s < T$, and  $p:= \lceil s/\Delta_{\text{graph}} \rceil$, calculate the conditional-least-squares estimates
\begin{align}
 \widehat{\bf{H}}_j^{{(\Delta_{\text{graph}},s)}}:=  \frac{1}{\Delta_{\text{graph}}} \left(\bfZ_j^\top \bfZ_j\right)^{-1} \bfZ_j^\top \bfY_j\in \R^{(pd_j +1)\times 1} ,\quad   j\in[d_j],\label{H_j}
 \end{align}
with design matrices
\begin{align}
\bfZ_j &:=
\left(\begin{array}{ccccc}(\mathbf{X}^{(\Delta_{\text{graph}})}_{p,\PA(j)})^\top & (\mathbf{X}^{(\Delta_{\text{graph}})}_{p-1,\PA(j)})^\top& \dots   & (\mathbf{X}^{(\Delta_{\text{graph}})}_{1,\PA(j)})^\top &1 \\
(\mathbf{X}^{(\Delta_{\text{graph}})}_{p+1,\PA(j)})^\top & (\mathbf{X}^{(\Delta_{\text{graph}})}_{p,\PA(j)})^\top & \dots  & (\mathbf{X}^{(\Delta_{\text{graph}})}_{2,\PA(j)})^\top & 1\\
 \\\dots  & \dots  & \dots  & \dots  &\dots \\
 \\
(\mathbf{X}^{(\Delta_{\text{graph}})}_{n - 1,\PA(j)})^\top & (\mathbf{X}^{(\Delta_{\text{graph}})}_{n-2,\PA(j)})^\top & \dots    & (\mathbf{X}^{(\Delta_{\text{graph}})}_{n-p,\PA(j)})^\top &1
\end{array}\right)\in\N_0^{ (n-p) \times (pd_j+1)}\label{Z_j},\quad j\in[d],
\end{align}
and vectors of responses
$$
\bfY_j :=\left(\mathbf{X}^{(\Delta_{\text{graph}})}_{p+1,j},\mathbf{X}^{(\Delta_{\text{graph}})}_{p+2,j},\dots ,\mathbf{X}^{(\Delta_{\text{graph}})}_{n,j}\right)^\top\in\N_0^{(n-p)\times1 },\quad j\in[d].
$$
Given  $\widehat{\bf{H}}_j^{{(\Delta_{\text{graph}},s)}}$, $j\in [d],$ we define the \emph{Hawkes-graph estimator} $\widehat{G}_{\bfN}:=(\widehat{\mathcal{V}}_{\bfN},\widehat{\mathcal{E}}_{\bfN})$ with
$\widehat{\mathcal{V}}_{\bfN} := \{(j;\hat{\eta}_j):\ j\in[d]\}$ and
\begin{align}
\widehat{\mathcal{E}}_{\bfN} :=\bigcup\limits_{j = 1,\dots,d} \Big\{(i_l,j; \hat{a}_{i_l,j}):\ \{i_1,\dots,i_{d_j}\} = {\PA}(j) ,\,\hat{a}_{i_l,j} =b_{l,j}^\top\widehat{\bf{H}}
^{{(\Delta_{\text{graph}},s)}}_j  \Big\},\label{H_j_calculation}
\end{align}
where, for $l\in[d_j]$, $b(l,j)\in\{0,1\}^{(d_jp + 1)\times 1}$ is a column vector with $0$s in all components but 1s in components $((k - 1)d_j + l),\ k=1,2,\dots,p$. Furthermore, for $\alpha_{\text{graph}}\in(0,1)$, we define the confidence intervals
$\big[\hat{\eta}_j\pm \hat{\sigma}_{j}z^{-1}_{1-\alpha_{\text{graph}}}\big),\, j\in[d],$ and, for $i_l\in\PA_{\bfN}(j)$,
$\big[\hat{a}_{i_l,j}\pm \hat{\sigma}_{i_l,j} z^{-1}(1-\alpha_{\text{graph}})\big)$. We give the calculation of $
\hat{\sigma}_{i_l,j}$ and $\hat{\sigma}_{j}$ in Algorithm~\ref{sigma_comp2}, below.
\end{definition}
As before, additional notation clarifies what the entries of the matrices $\widehat{\bf{H}}^{{(\Delta_{\text{graph}},s)}}_j,\, j\in [d],$ actually estimate:
\begin{align}
\left(
\begin{array}{c}
\widehat{H}_{\PA(j),j}(\Delta_{\text{graph}})  \\
\widehat{H}_{\PA(j), j}(2\Delta_{\text{graph}})  \\
 \dots \\
 \widehat{H}_{\PA(j), j}(p\Delta_{\text{graph}}) \\
 \hat{\eta}_j
 \end{array}
 \right)
 &:=\widehat{\bf{H}}_j,\ \text{with}\label{estimator_interpretation2}   \\
 \widehat{H}_{\PA(j), j}(k\Delta_{\text{graph}}) &= \(\hat{h}_{i_1, j}(k\Delta_{\text{graph}}),\hat{h}_{i_2, j}(k\Delta_{\text{graph}}), \dots,\hat{h}_{i_{d_j}, j}(k\Delta_{\text{graph}})\)^\top,
 \nonumber
\end{align}
$k = 1,2,\dots, p$ and $\{i_1, i_2, \dots ,i_{d_j}\} = \PA(j)$. 
%
Finally, we provide efficient computations for the covariance estimates that are necessary for the confidence intervals around the estimated edge and vertex weights.
\begin{algorithm}
\label{sigma_comp2}
Let $j\in[d]$ such that $|\PA(j)|>0$ and let $\{i_1,i_2,\dots,i_{d_j}\}=\PA(j)$ with $i_1<i_2<\dots,<i_{d_j}$.
For $(i_l,j),\, l \in [d_j],$ let $e(i_l,j)\in\{0,1\}^{(d_j p +1)\times 1}$ be a column vector with all entries 0, but 1s at components $(k-1)d_j + (l-1),\ k=1,2,\dots,p$. 
We compute $\hat{\sigma}_{i_l,j}$ in the following way:
\begin{enumerate}[label = \roman*)]
\item Compute $\mathbf{ C}_{l,j} := e(i_l,j)^\top \((\bfZ_j^\top \bfZ_j)^{-1} \bfZ_j^\top\)\in\R^{1\times (n-p)}$.
\item Set 
$
\bfU_j =  (\bfY_j - \Delta_{\text{graph}}  \bfZ_j  \widehat{\bfH}_j)\in\R^{(n-p)\times 1}$. Denoting $ (U_{p+1,j}, U_{p + 2,j},\dots, U_{n,j})^\top:= \bfU_j$, we have that
$$
U_{k,j} = \left(\bfX^{(\Delta_{\text{graph}})}_{k,j} - \Delta_{\text{graph}}\hat{\bfeta} - \sum\limits_{m=1}^p\Delta_{\text{graph}} \widehat{H}_{\PA(j),j}^\top(m\Delta_{\text{graph}}) \bfX^{(\Delta_{\text{graph}})}_{k-m,\PA(j)} \right),
$$
for $ k = {p+1}, p +2,\dots, n$.
\item 
Pointwise multiply $\mathbf{C}_{l,j}$ and $\bfU_j$. The sum of the squares of the result yields $\hat{\sigma}_{i_l,j}^2\in\R_{\geq 0}$.
 \end{enumerate}
For the variance estimates corresponding to the $j$-th vertex weight, 
consider the last row of $\((\bfZ_j^\top \bfZ_j)^{-1} \bfZ_j\)\in\R^{(d_j p + 1)\times (n-p)}$, multiply it pointwise with $\bfU_j$ from above, take the sum of squares of the results and multiply the result with $\Delta_{\text{graph}}^{-2}$; this yields $\hat{\sigma}_j^2$.
 \end{algorithm}

\begin{remark}
The bin size $\Delta_{\text{graph}}$ for the graph estimation in Definition~\ref{def:graph_estimator} will typically be much smaller than the bin size $\Delta_{\text{skel}}$ for the skeleton estimation in Definition~\ref{skeleton_estimator}. After the graph estimation, one might again want to delete edges with edge-weight estimates non-significantly different from zero, or treat vertex-weight estimates, respectively, immigration intensities, that are not significantly different from zero as zero; see Figure~\ref{fig2}. Also note that the latter could possibly be tested with a different significance parameter $\alpha_{\text{vertex}}$ than the significance parameter $\alpha_{\text{graph}}$ from the edge weight estimation. In any case, the resulting Hawkes-graph estimations ought to be checked for \emph{redundant vertices}; see Definition~\ref{def:graph}. If the estimate has redundant vertices, the results are typically inconsistent with the data---as we typically observe data in all components. Therefore, if a fitted model has redundant vertices, we ought to increase $\alpha_{\text{skel}}$, $\alpha_{\text{graph}}$, and/or $\alpha_{\text{vertex}}$. Thus, we obtain more estimated nonzero immigration intensities and/or larger estimated edge sets. We proceed with increasing the significance parameters until there are no redundancies left.  
\end{remark}
Given a Hawkes-graph estimate as in Definition~\ref{def:graph_estimator}, one may examine connectivity issues, path weights, graph distances, feedback and cascade coefficients, exploit graphical representations, etc.; see the example in Section~\ref{example}.

\subsection{Estimation of the reproduction intensities}\label{excitement_function_estimation}
For many applications, the results discussed above may already suffice. In other applications however, the graph estimation will only be a preliminary step and one would like to examine how the various excitements are distributed \emph{over time}. In other words, one would like to explicitly estimate the displacement intensities, respectively, the reproduction intensities from Definition~\ref{hawkes_family}.
\subsubsection*{Parametric estimation}  Given the Hawkes estimator from Definition~\ref{estimator},
the Hawkes model is not yet completely specified. In particular, \eqref{H_j} only yields estimates of the reproduction intensities on a grid:
\begin{align}
\bigg\{\Big(k\Delta\Big),\hat{h}_{i,j}(k\Delta)\bigg\}_{k = 1,2,\dots, p},\quad i\in\widehat{\PA}(j),\ j\in[d].\label{estimates}
\end{align}
One obvious possibility to complete the model specification would be the application of any kind of smoothing method on \eqref{estimates}. We want to consider another approach: we exploit \eqref{estimates} graphically (examine log/log-plots, id/log-plots, check for local maxima, convex/concave regions, etc.) and 
identify appropriate parametric families.  The parameters can then be fitted to the estimates \eqref{estimates} via non-linear least-squares (e.g., function {\tt nls} in {\tt R}):
\begin{definition}\label{def:parametric_estimation}
Consider a Hawkes-graph estimation as in Definition~\ref{def:graph_estimator} with respect to some $d$-type event-stream data and a bin size $\Delta_{\text{graph}}>0$. For $j\in[d]$ and $i\in\widehat{\PA}(j)$, let $w_{i,j}^{(\theta_{i,j})}:\,\R\to\R_{\geq 0},\, w_{i,j}^{(\theta_{i,j})}(t) = 0,\, t\leq 0$, be density families parametrized by $\theta_{i,j}\in \Theta_{i,j}\subset \R^{d_{i,j}}$. With the notation from \eqref{estimates}, let
\begin{align}
(\hat{a}_{i,j},\hat{\theta}_{i,j}) := \argmin\limits_{(a,\theta)\in \R_{\geq 0}\times \Theta_{i,j}}\sum\limits_{k = 1}^p \Big(a w_{i,j}^{(\theta)}\big(k \Delta_{\text{graph}}\big) -\hat{h}_{i,j}(k\Delta_{\text{graph}}) \Big)^2,\quad (i,j)\in \widehat{\mathcal{E}}^*, \label{eq:theta_ij}
\end{align}
and define the \emph{parametric reproduction-intensity estimates}
$$
\hat{h}^{\text{(par)}}_{i,j}(t):=\begin{cases}
\hat{a}_{i,j}{w}^{(\hat{\theta}_{i,j})}_{i,j}(t),& (i,j)\in \widehat{\mathcal{E}}^*,\quad t\in\R,\\
0,  &(i,j)\notin \widehat{\mathcal{E}}^*,\quad t\in\R,
\end{cases} 
$$
the \emph{parametric branching-matrix estimate}
$$
\widehat{A}^{\text{(par)}} := \( \int\hat{h}^{\text{(par)}}_{i,j}(t)\d t\)_{1\leq i, j \leq d},
$$
and the \emph{parametric immigration-intensity estimates}.
\begin{align}
\hat{\eta}^{\text{(par)}} := \(\hat{\eta}^{\text{(par)}}_1,\dots,\hat{\eta}^{\text{(par)}}_d\):=
\lambda^{\text{(emp)}}\(1_{d\times d} -\widehat{A}^{\text{(par)}}\),\label{eta_par_est}
\end{align}
where $\lambda^{\text{(emp)}}$ denotes the observed empirical intensity  $\lambda^{\text{(emp)}}:=\bfN\big((0,T]\big)/T\in\R_{\geq 0}^{1\times d}$. \end{definition}
We illustrate this specification and estimation of a fully parametric multivariate Hawkes process in Figure~\ref{fig3}. Here, we also see that the parameter estimates from \eqref{eq:theta_ij} are symmetrically distributed around the true values. Even though the estimator calculations in Definition~\ref{def:parametric_estimation} stand at the end of a long chain of various discretizations and truncations, `log-likelihood profile' confidence intervals (e.g., from {\tt confint.nls} in {\tt R}) give remarkably good coverage rates for the parameter estimates (not illustrated).
\begin{remark}
The definition of $\eta^{\text{(par)}}$ in $\eqref{eta_par_est}$ is motivated by the desirable equality
$$
 \eta^{\text{(par)}} \(1_{d\times d} -(\widehat{A}^\text{(par)})^\top\)^{-1} =\lambda^{\text{(emp)}}.
$$
In other words, with this choice of $\hat{\eta}^{\text{(par)}}$, the observed unconditional intensity exactly equals the estimated unconditional intensity. This might be relevant in some applications (e.g., simulation from a fitted model). Finally note that it might often be more efficient to consider weighted least squares in \eqref{eq:theta_ij}. 
 \end{remark}
\section{Example}\label{example}

We illustrate the concepts introduced in the previous sections with a ten-dimensional Hawkes model. We perform a simulation study and apply the estimation methods from Sections~\ref{estimation_of_the_hawkes_skeleton}, \ref{estimation_of_the_hawkes_graph}, and \ref{excitement_function_estimation}  to the Hawkes skeleton, the Hawkes graph, and the reproduction-intensity parameters.
\subsection{Example model}\label{example_model}
We consider a $10$-type Hawkes process $\bfN$ as in Definition~\ref{def:hawkes_process}  with immigration intensities 
\begin{align}\label{immigration_intensities}
\eta_i: =
\begin{cases}
1, & i\in\{1,7,10\},\\
0, & i\in\{2,3,4,5,6,8,9\},
\end{cases}
\end{align}
and reproduction intensities $h_{i,j},\, (i,j)\in[10]^2$, defined, for $t\in\R$, by
\begin{align}\label{reproduction_intensities}
\quad h_{i,j}(t) :=
\begin{cases}
1.5\,\gamma(t),& (i,j)\in\{(1,2),(2,4),(8,9)\},\\
 1_{t\in[1,2]}0.5,& (i,j)\in\{(1,1),(2,3),(3,5),(4,3), (4,5),(4,6),(5,3),(7,8),(9,7)\},\\
 1_{t\in[1,2]}0.1,&(i,j) = (5,7),\\
0,& \text{else}.
\end{cases}
\end{align}
Here, $\gamma$ denotes a Gamma density with shape parameter 6 and rate parameter 4, i.e.,
$\gamma(t)=1_{t\geq 0}t^{5} \exp\{-4t\} (4^6)/(5!)$. In Hawkes graph terminology, we have 13 edges supplied with three different kinds of edge weights: a \emph{heavy weight} (1.5) for three edges, a \emph{light weight} (0.5) for seven edges, and one edge with a \emph{super-light weight} (0.1). An illustration of the corresponding graph $\mathcal{G}_{\bfN}$ is much more meaningful than \eqref{reproduction_intensities}; see the left graph in Figure~\ref{fig1}. From this figure, the various direct and indirect dependencies can be read off instantaniously; only the large nodes have nonzero immigration intensity; a fat edge corresponds to an edge weight of 1.5; a thin edge corresponds to an edge weight of 0.5; the dashed line corresponds to the super-light edge weight 0.1. We examine the Hawkes-graph properties introduced in Definitions~\ref{hg}~and~\ref{coefficients}:\\
\vspace{0.5cm}\strut
\\
{\bf Redundancy} The Hawkes graph $\mathcal{G}_{\bfN}$ has no \emph{redundant vertices}: all small vertices have a large vertex as one of their ancestors. If vertex 1 were small, the vertices $1,2,3,4,5$ and $6$ would be \emph{redundant} as they could not generate events.\\

\noindent {\bf Connectivity} The Hawkes graph $\mathcal{G}_{\bfN}$ is \emph{not weakly connected}. The graph can be divided in two separate weakly-connected Hawkes subgraphs with vertex sets $\{1,2,3,4,5,6,7,8,9\}$, and 
$\{10\}$. Deleting edge $(5,7;0.1)$ would yield three separate weakly-connected Hawkes subgraphs.\\

\noindent{\bf Criticality} The Hawkes graph $\mathcal{G}_{\bfN}$ is \emph{subcritical}:
all vertices but vertex 10 are part of closed walks. It suffices to check criterion\ \eqref{subcriticality} for vertices $i_0 \in\{1,2,3, 7\}$.
For vertex $1$, we find that 
$$
\mathcal{W}_g^{(1,1)} = \{(\underbrace{1,1,\dots,1}_{g + 1 \text{ times}})\},\ g\in \N,
\quad
\text{and}\quad  |(\underbrace{1,1,\dots,1}_{g + 1 \text{ times}})| = 0.5^g,\, g\in\N.
$$
Consequently, 
$
\sum_{g = 1}^\infty \sum_{w_g\in\mathcal{W}_g^{(1,1)}} |w_g|= \sum_{g = 1}^\infty 0.5^g<\infty
$. For vertex $2$, we find that 
$$
\mathcal{W}_1^{(2,2)} =\mathcal{W}_2^{(2,2)} =  \emptyset,\, 
\mathcal{W}_3^{(2,2)}= \{(2,4,6,2)\}, \,\mathcal{W}_4^{(2,2)} =\mathcal{W}_5^{(2,2)} =  \emptyset, \,\mathcal{W}_6^{(2,2)}= \{(2,4,6,2,4,6,2)\}, \dots
$$
With $|(2,4,6,2)| = 1.5\cdot 0.5\cdot 0.5 = 0.375$, $|(2,4,6,2,4,6,2)| = 0.375^2,\,\dots,$ criterion \eqref{subcriticality} again follows. For vertices 3 and 7, one argues analogously. In other words, as long as closed walks do not overlap, we can construct large subcritical Hawkes graphs without calculating any eigenvalues. When closed walks overlap, the underlying combinatorics typically become too involved as to proceed in this manner. In this case one could calculate the spectral radius of the adjacency matrix of the involved edges only. For example, if we wanted to introduce another edge $(9,5; a_{9,5})$ in model \eqref{reproduction_intensities}, respectively, Figure~\ref{fig1}, we would have to calculate the spectral radius of the adjacency matrix corresponding to the Hawkes (sub-)graph with edges 
$$ 
\big\{(3,5;0.5),(5,3;0.5),(5,7;0.1),(7,8;0.5),(8,9;0.5),(9,5;a_{9,5}),(9,7;0.5)\big\};
$$
see Theorem~\ref{prop:graph_subcriticality}.\\

\noindent{\bf Cascade and feedback coefficients}
We calculate the coefficients from Definition~\ref{coefficients} with respect to the example model; see Table~\ref{table0}. 
The cascade and feedback coefficients summarize the impact of the driving vertices 1, 4 and 10 (that is, of the vertices with nonzero vertex weights) on the process. The {\it cascade coefficients} measure the impact of each vertex on the whole system. In our example, the immigrants in the first vertex together with the cascades that they trigger are responsible for about 82\% of all events that occur in the system. The {\it feedback coefficients} measure the impact of the impact of each vertex on itself. In our example, for vertex 8 this means that 76\% of its activity are explained by its own immigration activity and by the feedback loops that the immigrants possibly trigger via closed walks. Vertex 1 is only excited by its own activity. For vertex 10 the feedback coefficient is also equal 1---albeit there is no true feedback involved. Still, its intensity would decrease by 100\% if it were switched off.

\begin{table}[ht]
\caption{Cascade and feedback coefficients}
\label{table0}
\footnotesize
\centering
\begin{tabular}{rrrrrrrrrrr}
  \hline
 & 1 & 2 & 3 & 4 & 5 & 6 & 7 & 8 & 9 & 10 \\ 
  \hline
cascade.coefficients & 0.82 & 0.00 & 0.00 & 0.00 & 0.00 & 0.00 & 0.14 & 0.00 & 0.00 & 0.04 \\ 
  feedback.coefficients & 1.00 & 0.00 & 0.00 & 0.00 & 0.00 & 0.00 & 0.76 & 0.00 & 0.00 & 1.00 \\ 
   \hline
\end{tabular}
\end{table}


\subsection{Simulation study}\label{simulation_study}

\begin{table}
\footnotesize
\begin{center}
\caption{$\Delta_{\text{skel}} = 0.2$}
\label{table1}
\begin{tabular}{rrrrrrr}
  \hline
alpha.skel & nedges & total & heavy & light & super.light & zero \\ 
  \hline
0.005 & 12.324 & 0.902 & 1.000 & 0.956 & 0.121 & 0.993 \\ 
  0.010 & 13.066 & 0.917 & 1.000 & 0.970 & 0.190 & 0.987 \\ 
  0.050 & 17.296 & 0.946 & 1.000 & 0.990 & 0.379 & 0.942 \\ 
  0.100 & 21.995 & 0.959 & 1.000 & 0.995 & 0.507 & 0.890 \\ 
  0.250 & 35.015 & 0.979 & 1.000 & 0.999 & 0.739 & 0.744 \\ 
   \hline
\end{tabular}

\end{center}
\end{table}

\begin{figure}
\begin{center}
\includegraphics[width = 0.8\textwidth]{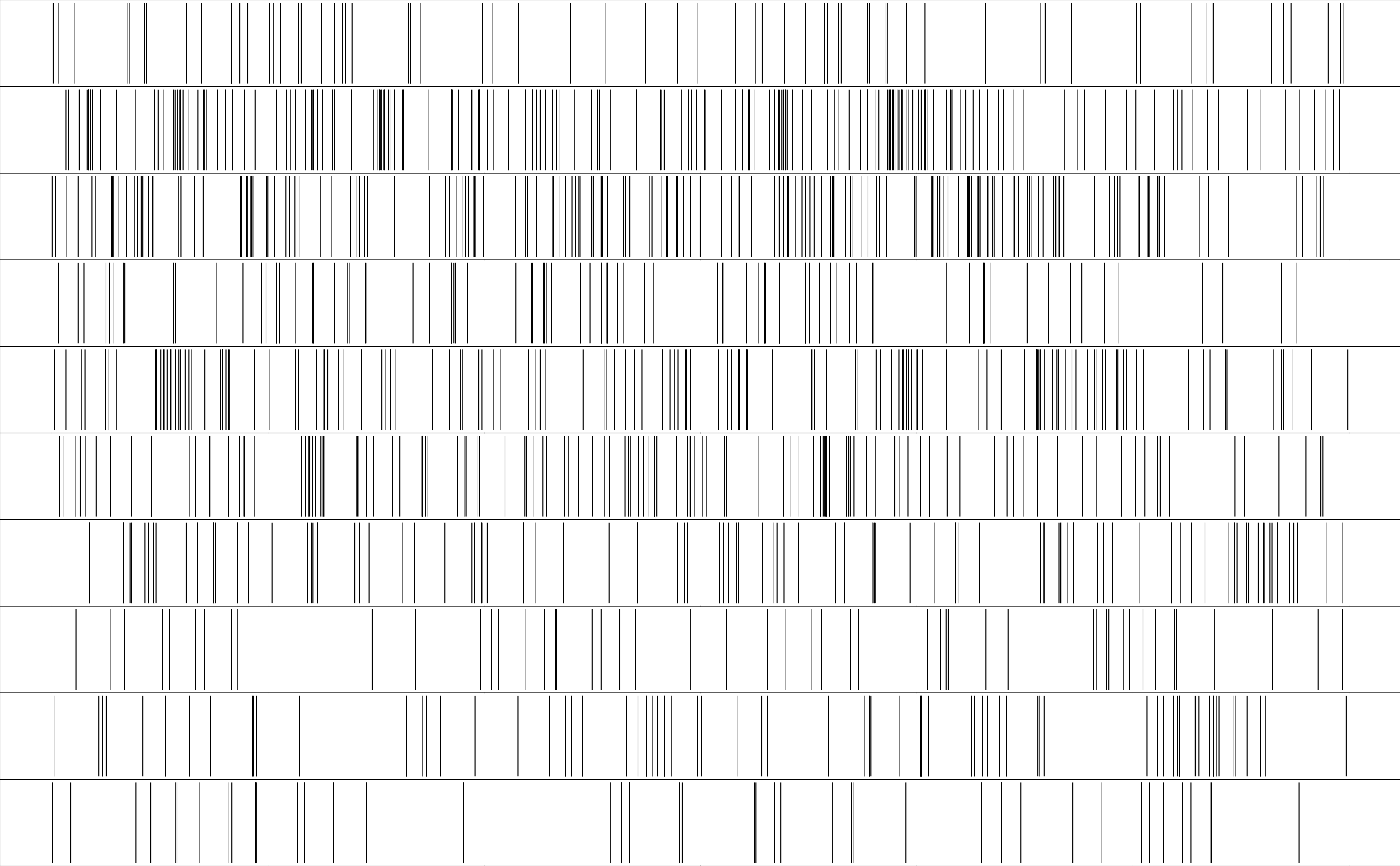}
\includegraphics[width = 0.48\textwidth]{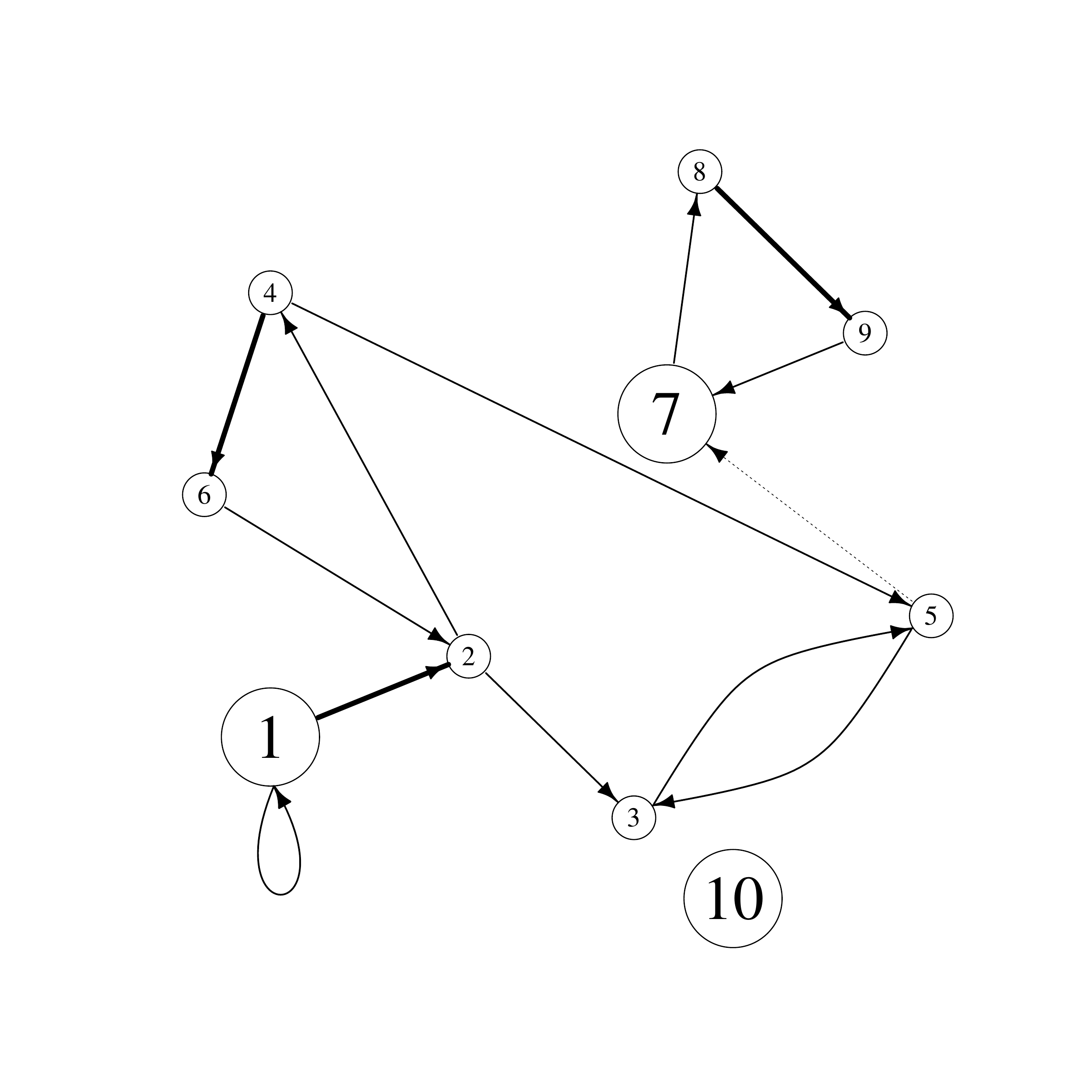}
\includegraphics[width = 0.48\textwidth]{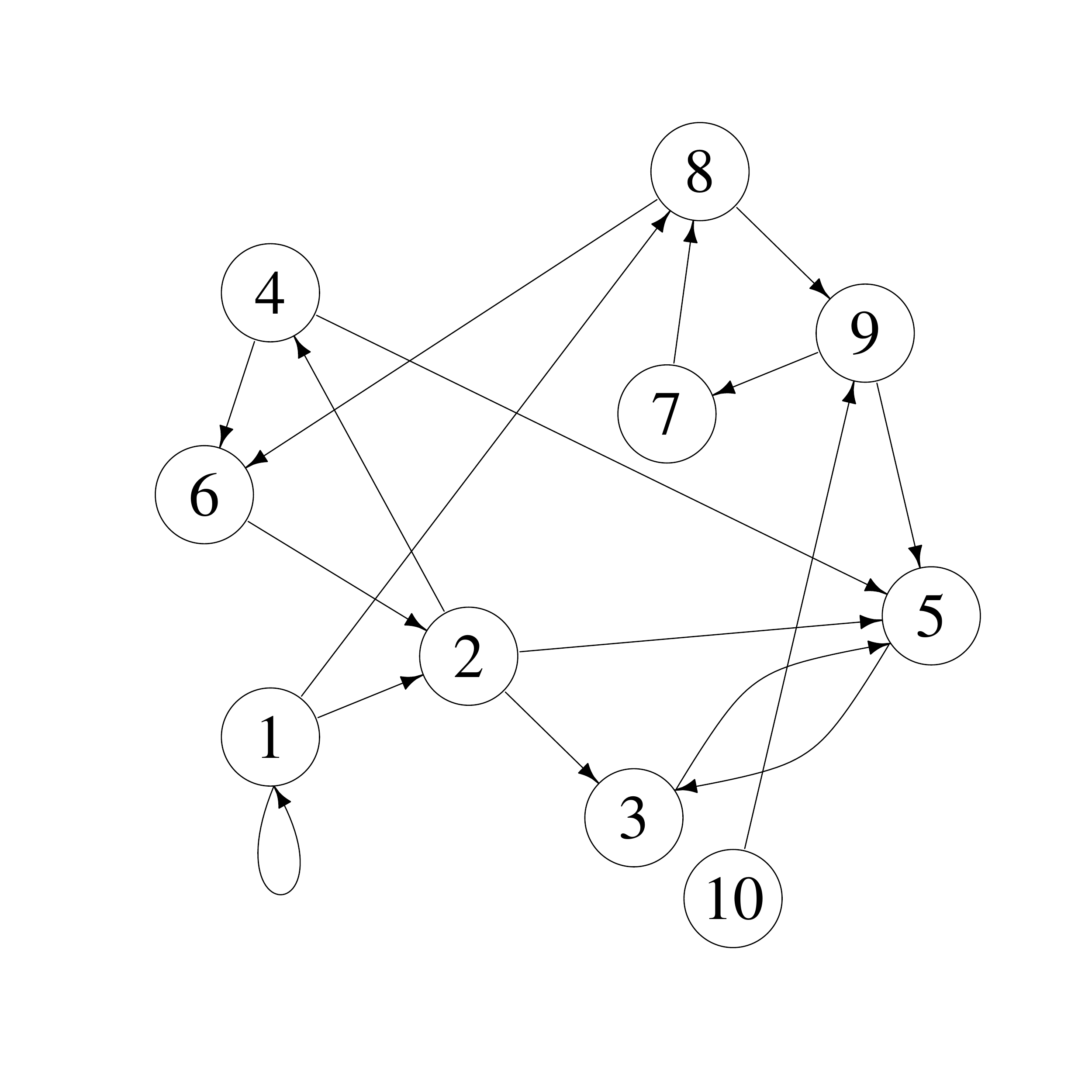}
\caption{Hawkes process simulation, Hawkes graph, and estimated Hawkes skeleton.\
The left graph represents the Hawkes graph corresponding to the Hawkes process example from Section \ref{example_model}; the graph is a summary of the immigration and branching structure of the model: edges from one vertex to another vertex denote nonzero reproduction intensities, respectively, excitement. Fat edges refer to heavy excitement (1.5 expected children events in branching construction); thin edges to small excitement (0.5 expected children) and the dotted line refers to a very small excitement (0.1 expected children); see~\eqref{reproduction_intensities}. Large vertices correspond to nonzero immigration-intensities ($=1$) and small vertices to the zero-immigration vertices; see \eqref{immigration_intensities}. The barcode plots illustrate a 30 time-units window of a simulated realization of the model (after some burn-in): we observe events of ten types, respectively, in ten components. One goal of our paper is to retrieve the graph on the left from such a realization. As a first step towards this aim, we calculate the Hawkes-skeleton estimate from Definition~\ref{def:graph_estimator} with respect to a coarse bin size $\Delta_{\text{skel}} = 1$ and a sparseness parameter $\alpha_{\text{skel}} = 0.05$. The right graph illustrates such an estimate. This skeleton will be used in a second step to retrieve the Hawkes-graph estimate; see Figure~\ref{fig2}. Comparing the skeleton with the true graph on the right, we see that we catch twelve of the thirteen true edges. We miss edge $(5,7)$. Furthermore, the estimate introduces five additional wrong edges $(1,8)$, $(2,5)$,$(8,6)$, $(9,5)$, and $(10,9)$. The three crucial points are: (i) These five false-positive edges \emph{do not introduce additional bias} in the graph estimation. (ii) Due to the coarse $\Delta_{\text{skel}}$-value, the calculation of the skeleton estimate is computationally simple. (iii) The resulting skeleton estimate is nearly as sparse as the true skeleton. This considerably reduces the complexity of the graph estimation (with a very fine $\Delta_{\text{graph}}$-parameter). See Figure~\ref{fig2}, for the Hawkes-graph estimation with respect to the skeleton estimate from above.}\label{fig1}
\end{center}
\end{figure}

We simulate $n_{\text{sim}} = 1000$ realizations of the Hawkes process $\bfN$ from Section~\ref{example_model}. We use the branching construction from Definitions~\ref{hawkes_family} and~\ref{def:hawkes_process} as simulation algorithm. In each realization, we simulate a time window of $500$ time units. This typically yields between 500 and 2000 events per component. Given each of these realized event streams, we calculate the Hawkes-skeleton estimator from Definition~\ref{skeleton_estimator}---with respect to different values of $\Delta_{\text{skel}}$ and $\alpha_{\text{skel}}$. Given these skeleton estimates, we calculate the Hawkes-graph estimator from Definition~\ref{def:graph_estimator}---including confidence bounds for all vertex and edge weights. Finally, we analyze the scatterplots for branching-intensity estimates, choose parametric function families, and fit the parameters on the estimates by nonlinear least squares. Figures~\ref{fig1} and~\ref{fig2} illustrate the procedure.

\begin{table}
\footnotesize
\begin{center}
\caption{$\Delta_{\text{skel}} = 0.5$}
\label{table2}
\begin{tabular}{rrrrrrr}
  \hline
alpha.skel & nedges & total & heavy & light & super.light & zero \\ 
  \hline
0.005 & 12.353 & 0.902 & 1.000 & 0.957 & 0.120 & 0.993 \\ 
  0.010 & 13.118 & 0.917 & 1.000 & 0.971 & 0.179 & 0.986 \\ 
  0.050 & 17.255 & 0.945 & 1.000 & 0.990 & 0.375 & 0.943 \\ 
  0.100 & 21.952 & 0.959 & 1.000 & 0.995 & 0.514 & 0.891 \\ 
  0.250 & 34.805 & 0.980 & 1.000 & 0.999 & 0.745 & 0.746 \\ 
   \hline
\end{tabular}

\end{center}
\end{table}

\begin{table}
\footnotesize
\begin{center}
\caption{$\Delta_{\text{skel}} = 1$}
\label{table3}
\begin{tabular}{rrrrrrr}
  \hline
alpha.skel & nedges & total & heavy & light & super.light & zero \\ 
  \hline
0.005 & 12.476 & 0.910 & 1.000 & 0.967 & 0.129 & 0.993 \\ 
  0.010 & 13.171 & 0.921 & 1.000 & 0.977 & 0.178 & 0.986 \\ 
  0.050 & 17.264 & 0.949 & 1.000 & 0.993 & 0.400 & 0.943 \\ 
  0.100 & 21.806 & 0.962 & 1.000 & 0.997 & 0.535 & 0.893 \\ 
  0.250 & 34.465 & 0.979 & 1.000 & 0.999 & 0.730 & 0.750 \\ 
   \hline
\end{tabular}

\end{center}
\end{table}

\begin{table}
\footnotesize
\begin{center}
\caption{$\Delta_{\text{skel}} = 2$}
\label{table4}
\begin{tabular}{rrrrrrr}
  \hline
alpha.skel & nedges & total & heavy & light & super.light & zero \\ 
  \hline
0.005 & 12.244 & 0.810 & 1.000 & 0.828 & 0.074 & 0.980 \\ 
  0.010 & 13.680 & 0.846 & 1.000 & 0.876 & 0.119 & 0.969 \\ 
  0.050 & 19.709 & 0.913 & 1.000 & 0.957 & 0.262 & 0.910 \\ 
  0.100 & 25.065 & 0.936 & 1.000 & 0.978 & 0.369 & 0.852 \\ 
  0.250 & 38.186 & 0.966 & 1.000 & 0.994 & 0.605 & 0.705 \\ 
   \hline
\end{tabular}

\end{center}
\end{table}
\subsubsection*{Hawkes-skeleton estimation} We fix $s = 5$ and, for each simulated event-stream, we calculate the Hawkes-skeleton estimates from Definition~\ref{skeleton_estimator} with respect to this support parameter $s$, bin sizes $\Delta_{\text{skel}}\in\{0.2, 0.5,1,2\}$, and various sparseness parameters $\alpha_{\text{skel}}\in\{0.005, 0.01, 0.05, 0.1, 0.25\}$. We denote the estimated edge sets by $\{\widehat{\mathcal{E}}^*(k)\}_{k = 1,2,\dots,n_{\text{sim}}}$ and the true edge set by ${\mathcal{E}}^*$. Using this notation, we summarize the results of the simulation study in Tables~\ref{table1},~\ref{table2},~\ref{table3}, and~\ref{table4} with the following statistics:
 \begin{enumerate}[label = \roman*)]
\item {\it nedges}:\  average size of estimated edge-sets (true number is 13), that is, $\sum_{k = 1}^{n_{\text{sim}}}|\mathcal{E}^*(k)|/n_{\text{sim}}$.
\item\label{item:total} {\it total}:\ fraction of correctly included edges, i.e, of pairs $(i,j)\in \widehat{\mathcal{E}}^*_{\bfN}(k)$ such that $(i,j)\in{\mathcal{E}}_{\bfN}$:
$$
\frac{\sum_{k = 1}^{n_{\text{sim}}}\sum_{(i,j)\in {\mathcal{E}}^*}1_{\{(i,j)\in\widehat{\mathcal{E}}^*(k)\}}}{n_{\text{sim}}|\mathcal{E}^*|}.
$$
Note that $1- total$ is the \emph{false-negative rate}.
\item {\it heavy/light/super.light}: more detailed version of \ref{item:total} above; fractions of correctly estimated edges with heavy (1.5), light (0.5) and super-light (0.1) edge weights.
\item {\it zero}: fraction of correctly excluded edges, i.e., of pairs $(i,j)\notin \widehat{\mathcal{E}}^*_{\bfN}(k)$ such that $(i,j)\notin{\mathcal{E}}_{\bfN}$:
$$
\frac{\sum_{k = 1}^{n_{\text{sim}}}\sum_{(i,j)\notin{\mathcal{E}}^*}1_{\{(i,j)\notin\widehat{\mathcal{E}}^*(k)\}}}{n_{\text{sim}}\big(d^2 -|\mathcal{E}^*|\big)}.
$$
Note that $1- zero$ is the \emph{false-positive rate}.
\end{enumerate}
First, we discuss the estimations with respect to bin size $\Delta_{\text{skel}} = 0.2$; see Table~\ref{table1}. We note from the last column, \emph{zero}, that the false-positive rate is indeed very close to the value of the chosen theoretical significance level $\alpha_{\text{skel}}$. Going back to Definition~\ref{skeleton_estimator}, we see that the larger $\alpha_{\text{skel}}$, the more edges are included in the Hawkes-skeleton estimation. This is reflected in all of the columns. However, even for very small $\alpha_{\text{skel}}$, we detect \emph{all} of the edges with a heavy edge weight and most of the edges with light edge weight. The edge $(5,7)$ with the super-light weight (0.1) is obviously a hard-to-detect alternative to the zero hypothesis. Note that Tables~\ref{table2}, \ref{table3}, and \ref{table4} look roughly the same as Table~\ref{table1} one above---though the estimates were calculated with respect to completely different bin sizes $\Delta_{\text{skel}}$.  So, in this first estimation step, we may use a very coarse bin size $\Delta_{\text{skel}}$. This makes the calculations underlying the skeleton estimation feasible even for much higher dimensions.\\
\par
The main purpose of the skeleton estimation is to lay the ground for the graph estimation which itself depends on a given estimated skeleton; see Definition~\ref{def:graph_estimator}. Missing edges in the skeleton estimate will typically introduce a bias for the graph-weight estimates. We therefore want to keep the false-negative rate ($ = 1 - \emph{total}$) in the skeleton estimation very small. As a consequence, we need $\alpha_{\text{skel}}$ large to include more edges.
Note that false-positive edges do \emph{not} add additional bias in the graph estimation; see Section~\ref{estimation_of_the_hawkes_graph}. So the increase of the false-positive rate (that is, the decrease in the \emph{zero}-column) does not prevent us from increasing the $\alpha_{\text{skel}}$-parameter. Note, however, that the whole reason  of the two-step estimation procedure is that in the first step we want to take advantage of the sparseness of the underlying true Hawkes graph and \emph{reduce} the complexity of the a priori fully connected network. Too many additional false-positive edges would hamper this advantage. In this sense, not only $\Delta_{\text{skel}}$ but also $\alpha_{\text{skel}}$ can be understood as a parameter controlling the numerical complexity of the method: the smaller $\alpha_{\text{skel}}$, the sparser the estimated skeleton, the less complex the computations for the Hawkes-graph estimate from Definition~\ref{def:graph_estimator}.
We see in our tables that, for all choices of $\Delta_{\text{skel}}$ and all values of $\alpha_{\text{skel}}$, we typically catch all the true edges, i.e., the false-negative rate is really small. In the next section, we will see that the graph estimates are not dramatically sensitive to the $\alpha_{\text{skel}}$ parameter in the skeleton estimation.

\subsubsection*{Hawkes-graph estimation}
In a further step, we quantify the estimated excitements. That is, given a Hawkes skeleton, we estimate the corresponding graph as in Definition~\ref{def:graph_estimator}; see Figure~\ref{fig2}.
\begin{figure}
\begin{center}
\includegraphics[width = 0.6\textwidth]{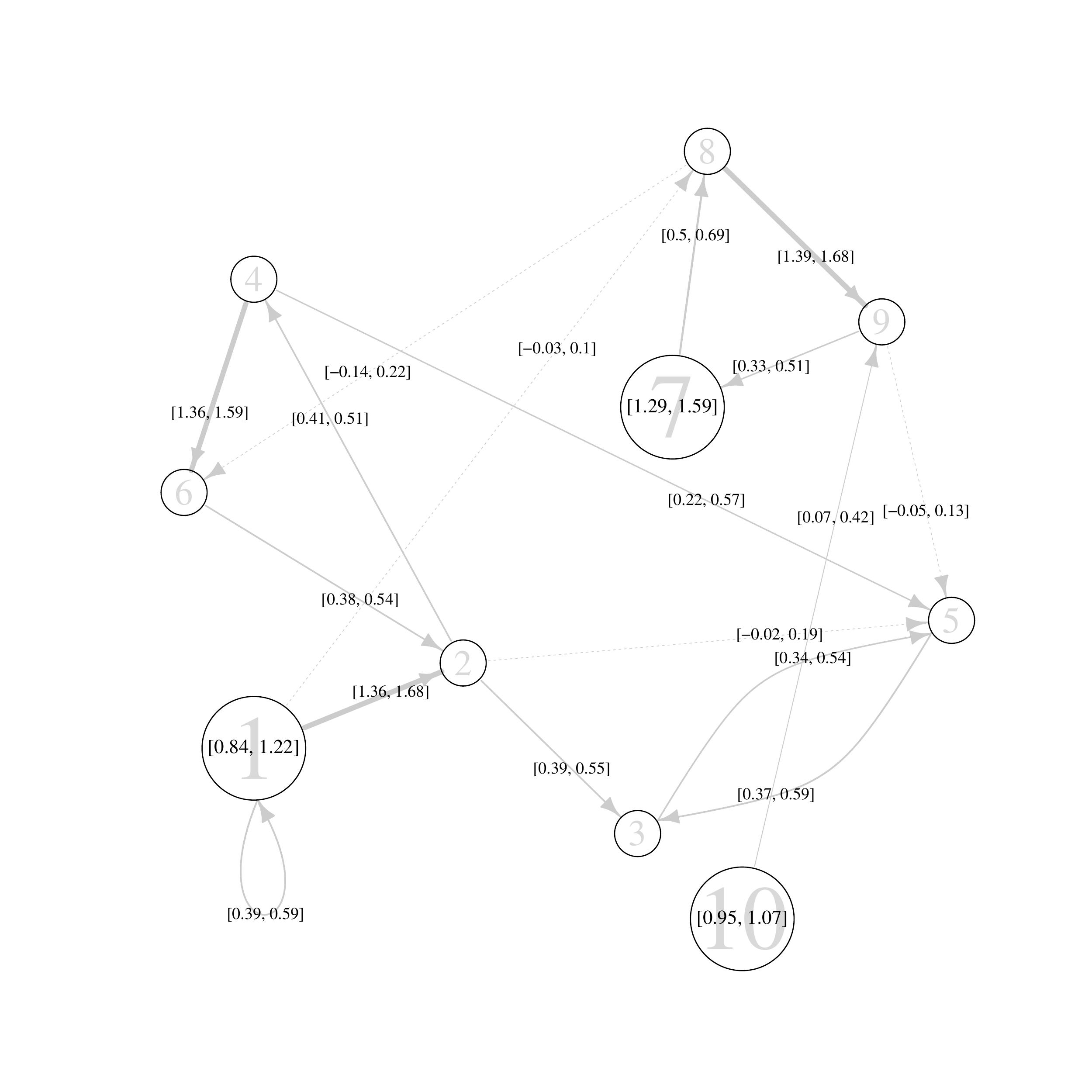}
\caption{Hawkes-graph estimation. Given a single simulation of length $T=1000$ from the example Hawkes process in Section~\ref{example_model}, we calculate the Hawkes-graph estimator from Definition~\ref{def:graph_estimator} with respect to the Hawkes-skeleton estimation from Figure~\ref{fig1}; we apply a bin size $\Delta_{\text{graph}} = 0.025$ and 
a significance parameter $\alpha_{\text{graph}}= 0.05$. This calculation allows us to supply each vertex and each node from this estimated skeleton with confidence intervals for their weights in the corresponding Hawkes graph. The edge widths in the illustration are chosen proportional to the estimated edge weights. Estimated edge weights that are not significantly larger than zero are illustrated as a dashed edge. Similarly, vertices where the confidence interval for the vertex weight contains 0 are plotted as smaller circles---the corresponding confidence bounds are left away in this latter case.
Comparing the results with the true Hawkes graph in Figure~\ref{fig1}, respectively, with the Hawkes process parametrization in \eqref{immigration_intensities} and \eqref{reproduction_intensities}, we see that for all correct edges, the true weights are covered by the confidence intervals. And for the wrong, additional edges from the skeleton estimation $(1,8)$, $(2,5)$,$(8,6)$, and $(9,5)$, we see that their weights are not significantly different from zero ($\alpha_{\text{graph}} = 0.05$). The estimated edge weight for the wrong edge $(10,9)$ is significantly larger than zero but still small. All true vertex weights but the weight of vertex 7 are also covered by the confidence intervals. The weight of vertex 7 is overestimated because we missed the (light) edge $(5,7;0.1)$ in the skeleton estimation; this missing explanatory variable for the events in component 7 is compensated by an extra large vertex weight in the graph estimation. Deleting all insignificant (in figure dashed) edges and setting the vertex weight of the insignificant (in figure small) vertex-weights to zero, we recover the original underlying graph almost perfectly. 
}\label{fig2}
\end{center} 
\end{figure}
We do this both with respect to the true skeleton and with respect to the estimated skeletons from the first estimation step. For comparison, we apply skeletons that were estimated with respect to different $\alpha_{\text{skel}}$-parameters. However, we only consider the skeletons that were estimated with respect to the (rough) bin size $\Delta_{\text{skel}} = 1$. As opposed to the skeleton estimation, we may now use a much smaller bin size  $\Delta_{\text{graph}} = 0.1$ for the graph estimation. In the present example, this is approximately the lower bin-size bound for tolerable computing time for the simulation study using a 2.3 GHz Intel Core processor (about 10sec for each of the estimations, no parallelization). 
Furthermore, we apply $s = 5$ and $\alpha_{\text{graph}} = 0.05$ in the calculation. For each simulation, we also calculate the confidence bounds for all vertex and edge weights from Definition~\ref{def:graph_estimator}. Table~\ref{coverage} reports the coverage rates.

\begin{table}\footnotesize
\begin{center}\caption{$\Delta_{\text{graph}} = 0.1$ and $\alpha_{\text{graph}} = 0.05$}
\begin{tabular}{lrr}
  \hline
applied.skeleton & vertex.weight.coverage & edge.weight.coverage \\ 
  \hline
alpha.skel = 0.005 & 0.859 & 0.907 \\ 
  alpha.skel = 0.01 & 0.867 & 0.904 \\ 
  alpha.skel = 0.05 & 0.896 & 0.893 \\ 
  alpha.skel = 0.1 & 0.907 & 0.900 \\ 
  alpha.skel = 0.25 & 0.915 & 0.932 \\ 
  true skeleton & 0.947 & 0.943 \\ 
   \hline
\end{tabular}
\label{coverage}
\end{center}
\end{table}
The coverage rates of the graph estimations that were calculated with respect to the true underlying skeleton correspond well with the significance parameter $\alpha_{\text{graph}}=0.05$. Naturally, the coverage rates for the estimates with respect to the estimated skeleton are smaller: as soon as the estimated skeleton misses an edge (e.g., the super-light edge $(5,7;0.1)$), the model calibration balances this missing possibility of excitement by increased baseline intensities or increased edge weights. The larger $\alpha_{\text{skel}}$, the lower the probabilty of missing an edge, the better the coverage rates. Note, however, that at the same time, the corresponding skeleton estimate becomes increasingly dense and with it the graph estimation becomes increasingly time-consuming.

\subsubsection*{Parametric reproduction intensity estimation}
Finally, we check how the various excitements are distributed over time. As examples, we examine the reproduction intensity $h_{1,2}$.
From the calculation of the Hawkes-graph estimate, we retrieve estimates of the reproduction intensity values on an equidistant grid; see \eqref{estimator_interpretation2}.
Based on the scatter plots of these estimates, we choose appropriate parametrized function families. Given such parametric functions, the parameters are fit to the pointwise estimates via nonlinear least squares; see Figure~\ref{fig3}. QQ-plots (not included) support asymptotic normality for the parameter estimates.
\begin{figure}
\begin{center}
\includegraphics[width = 0.4\textwidth]{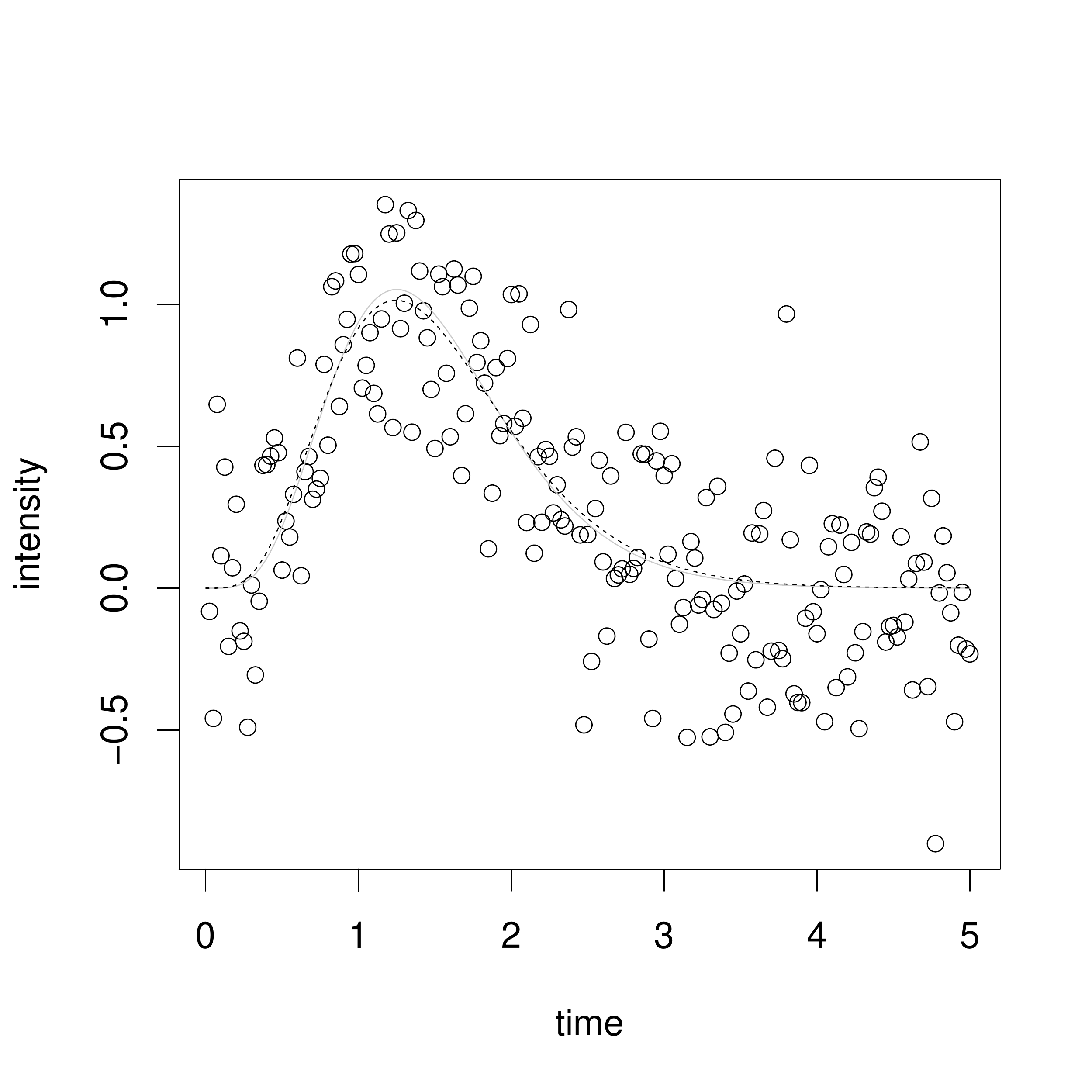}
\includegraphics[width = 0.4\textwidth]{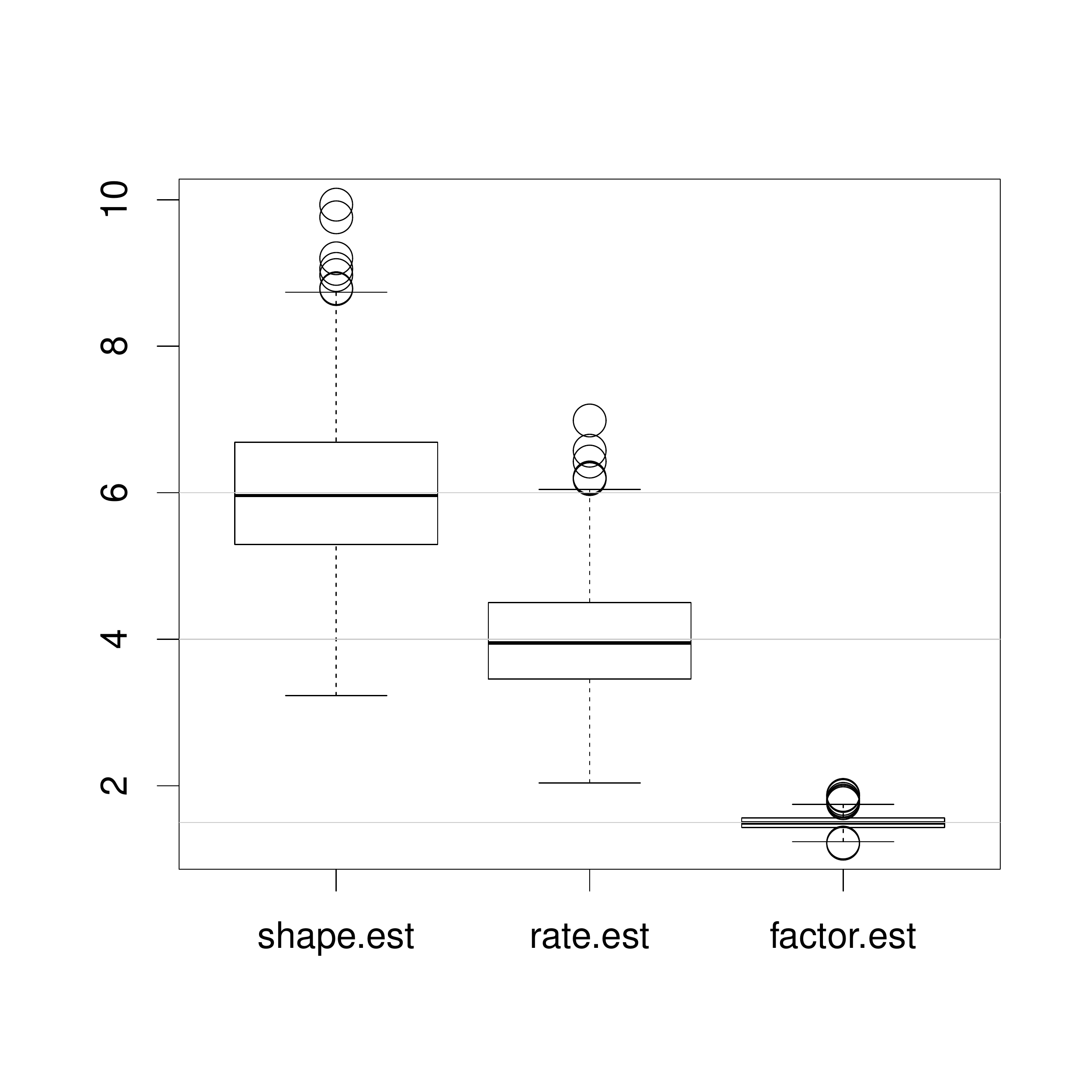}
\caption{Parametric estimation. \emph{Left:} From a single realization of the example Hawkes model from Section~\ref{example_model} with length $T = 1000$, we calculate Hawkes-skeleton and Hawkes-graph estimates from Definitions~\ref{skeleton_estimator} and~\ref{def:graph_estimator}; see Figures~\ref{fig1} and~\ref{fig2}. As a by-product of these calculations, we retrieve pointwise estimates (circles) for the values of the reproduction intensity $h_{1,2}$ on an equidistant grid; see \eqref{estimator_interpretation2}. From these estimates, one may guess that $h_{1,2}(t) = a_{1,2} \gamma(t)$, where $\gamma$ is a Gamma density depending on a shape and on a rate parameter. We fit the three parameters by nonlinear least squares as described in Section~\ref{excitement_function_estimation}. The dotted black line refers to the corresponding estimated parametric function. It catches the true underlying function (grey solid line) quite well; see \eqref{reproduction_intensities}. \emph{Right:} We apply this parametric estimation of $h_{1,2}$ on 1000 independent realizations of length $T = 500$. The boxplots collect the parameter estimates for each of the 1000 estimations of the simulation study. The grey marks refer to the corresponding true values. Eyeball-examination shows that the estimates are remarkably symmetric distributed and unbiased. QQ-plots (not illustrated) support asymptotic normality.
}\label{fig3}
\end{center}\end{figure}

\section{Conclusion}\label{conclusion}
The Hawkes graph and the Hawkes skeleton describe the immigration and branching structure of a Hawkes process in a graph-theoretical framework. We demonstrate how graph terminology can be very useful for multivariate Hawkes processes. Combining the new concepts with an estimation procedure from earlier work, we develop a statistical estimation method for the Hawkes skeleton and the Hawkes graph. The key idea is that in a preliminary step we only test if there is \emph{at all} excitement from any vertex to another vertex. We show that this first step is relatively simple to implement. The knowledge of the Hawkes skeleton makes the second step, the estimation of the Hawkes graph, much more efficient---both from a computational and statistical point of view. The simulation study shows that the procedure works as desired.  As long as the true underlying graph is sparse (e.g., if the typical number of parents of a node is not larger than 5 and does not depend on the dimension of the process) the approach may be applied in even higher-dimensional situations. In any case, the method may be a useful tool for preliminary analysis when examining large multi-type event-stream data in the Hawkes framework. \\
\par
It might be worthwile to study the distributional properties of the parameter estimates from Section~\ref{excitement_function_estimation} in more detail. Also note that the graph representation would also apply for discrete-time event-stream models, i.e., for multivariate time series of counts. More specifically, the present paper could have been developed in complete analogy for multivariate integer-valued autoregressive time series (INAR($\infty$)) which can be interpreted as discrete-time versions of the Hawkes process; see \cite{kirchner16c}. In this latter case, all  results that we apply in our paper would be valid without taking any discretization error into account. In any case, when applied to real data, the discretization error is \emph{not} the major drawback of our method: our method does indeed solve the important problem of how to decide whether an edge between two components exists at all. But for the specification of a Hawkes process we need to solve another---more important---issue. We want to be able to decide whether we observe a {\it complete} Hawkes graph or whether our data lack some non-redundant vertices! In particular, the method presented will also yield reasonable results for data stemming from models with no or less underlying `causality'. The seeming excitement can then be explained by a confounding factor that we do not observe (and ignore). We believe, in view of the widespread interpretation of the Hawkes model as a causal model (an interpretation we share), it would be of utmost importance to derive tests for the presence of such hidden confounding factors in the event-stream context.

\section*{Acknowledgements}
This research has been supported by the ETH RiskLab and the Swiss Finance Institute.
The authors wish to express their gratitude to all the {\tt R}-programmers providing and maintaining powerful statistical software. For our work, the igraph package \citep{csardi06} and the Matrix package \citep{maechler12} have been particularly helpful. We thank Vladimir Ulyanov for his comments on an earlier version of the paper which helped to improve the presentation. We also thank Philippe Deprez for a fertile discussion about Theorem~\ref{prop:graph_subcriticality}.
\bibliographystyle{apalike}
\bibliography{/Users/matthiaskirchner/Desktop/ETH/Diss/Bibliographies/diss.biblio}
\end{document}